\documentclass[11pt,a4paper,reqno]{amsart}%
\usepackage{amsthm,amsmath,amsfonts,amssymb,amsxtra,appendix,bookmark,dsfont,bm,mathrsfs}

\theoremstyle{plain}
\newtheorem{theorem}{Theorem}
\newtheorem{lemma}[theorem]{Lemma}

\newtheorem{proposition}[theorem]{Proposition}

\theoremstyle{definition}

\theoremstyle{remark}
\newtheorem{remark}{Remark}

\DeclareMathOperator{\Tr}{Tr}

\def\leqslant{\le}
\def\bq{\begin{eqnarray}}
\def\eq{\end{eqnarray}}
\def\bqq{\begin{align*}}
\def\eqq{\end{align*}}
\def\be{\begin{equation}}
\def\ee{\end{equation}}

\def\nn{\nonumber}

\def\wto{\rightharpoonup}
\newcommand{\norm}[1]{\left\lVert #1 \right\rVert}

\renewcommand{\epsilon}{\varepsilon}
\newcommand\1{{\ensuremath {\mathds 1} }}

\def\cF {\mathcal{F}}

\def\cN{\mathcal{N}}

\def\cL {\mathcal{L}}
\def\R {\mathbb{R}}

\def\C {\mathbb{C}}

\def\N {\mathcal{N}}

\def\cE {\mathcal{E}}

\def\H{\gH}

\def\R {\mathbb{R}}
\def\C {\mathbb{C}}
\def\N {\mathcal{N}}

\newcommand{\gH}{\mathfrak{H}}

\newcommand{\bH}{\mathbb{H}}
\newcommand{\dGamma}{{\ensuremath{\rm d}\Gamma}}

\allowdisplaybreaks

\title[Bogoliubov dynamics]{Bogoliubov correction to the mean-field dynamics of interacting bosons}

\author[P.T. Nam]{Phan Th\`anh Nam}
\address{Institute of Science and Technology Austria, Am Campus 1, 3400 Klosterneuburg, Austria} 
\email{pnam@ist.ac.at}
\author[M. Napi\'orkowski]{Marcin Napi\'orkowski}
\address{Institute of Science and Technology Austria, Am Campus 1, 3400 Klosterneuburg, Austria} 
\email{mnapiork@ist.ac.at}

\begin{document}
\date{\today}

\begin{abstract} We consider the dynamics of a large quantum system of $N$ identical bosons in 3D interacting via a two-body potential of the form $N^{3\beta-1} w(N^\beta(x-y))$. For fixed $0\leq \beta <1/3$ and large $N$, we obtain a norm approximation to the many-body evolution in the $N$-particle Hilbert space. The leading order behaviour of the dynamics is determined by Hartree theory while the second order is given by Bogoliubov theory. 
\end{abstract}

\maketitle

\setcounter{tocdepth}{1}
\tableofcontents

\section{Introduction}

We consider a large system of $N$ identical bosons living in $\R^3$ described by the Hamiltonian 
\begin{equation} 
H_N= \sum\limits_{j = 1}^N -\Delta_{x_j} + \frac{1}{N-1} \sum\limits_{1 \leqslant j < k \leqslant N} {w_N(x_j-x_k)},
\label{eq:def-HN}
\end{equation}
which acts on the symmetric space $\H^N=\bigotimes_{\text{sym}}^N L^2(\R^3)$. Here $x_j\in \R^3$ stands for the coordinate of the $j$-th particle. We assume that the interaction potential has the explicit form
\begin{equation} \label{eq:ass-wN}
w_N(x-y)= N^{3\beta} w(N^\beta (x-y))
\end{equation}
where $\beta \ge 0$ is a fixed parameter and $w$ a given nice function. For simplicity, we will assume that $w\in C^1_0(\R^3)$ is non-negative, spherically symmetric and decreasing.

Since the Hamiltonian $H_N$ is bounded from below, and hence it can be defined as a self-adjoint operator on $\gH^N$ by Friedrichs' method ~\cite{ReeSim2}. The coupling constant $1/(N-1)$ in front of the interaction is to ensure that the kinetic energy and interaction energy are of the same order. We could choose $1/N$ instead of $1/(N-1)$, but the latter will simplify some expressions. 
\smallskip

Note that when $\beta>0$, then $w_N$ converges to the Dirac-delta interaction. In general, the larger the parameter $\beta$, the harder the analysis. In this paper we will focus on the mean-field regime $0\le \beta<1/3$. In this case, the range of the interaction potential is much larger than the average distance between the particles and there are many but weak collisions. Therefore, to the leading order, the interaction potential experienced by each particle can be approximated by the mean-field potential $\rho*w_N$ where $\rho$ is the density of the system. If $\beta > 1/3$, then the analysis is expected to be more complicated due to strong correlations between particles.
\smallskip

In the present paper, we are interested in the large $N$ asymptotic behavior of the Schr\"odinger evolution 
\begin{equation} \label{eq:schrodingerdynamics}
\Psi_N(t) = e^{-itH_N}\Psi_{N}(0)
\end{equation}
generated by a special class of initial states $\Psi_{N}(0)\in \gH^N$. We are motivated by the physical picture that $\Psi_N(0)$ is a ground state (or an approximated ground state) of a trapped system described by the Hamiltonian
\begin{equation} \label{eq:HN-V} 
H_N^V= \sum\limits_{j = 1}^N (-\Delta_{x_j}+V(x_j)) + \frac{1}{N-1} \sum\limits_{1 \leqslant j < k \leqslant N} {w_N(x_j-x_k)}
\end{equation}
with an external potential $V\in L^\infty_{\rm loc}(\R^3,\R)$ satisfying $V(x)\to \infty$ as $|x|\to \infty$, and the time evolution $\Psi_N(t)$ in \eqref{eq:schrodingerdynamics} is observed when the trapping potential $V(x)$ is turned off. This leads to certain properties of $\Psi_N(0)$ which will be described below.

\subsection{Ground state properties}  
\subsubsection*{Bose-Einstein condensation} It is widely expected that ground states of trapped systems exhibit the (complete) Bose-Einstein condensation, namely $\Psi \approx u^{\otimes N}$ in an appropriate sense. In fact, when $0\le \beta<1$ we have 
\begin{equation}\label{eq:BEC-energy}
\lim_{N\to \infty}  \left(\inf_{\|\Psi \|_{\gH^N}=1} \frac{\langle \Psi, H_N^V \Psi \rangle}{N}  - \inf_{\|u\|_{\gH}=1}  \cE^{V}_{\rm{H},N}(u)\right)=0
\end{equation}
where
\begin{equation} \label{eq:def-EH}
\cE^{V}_{\rm{H},N}(u):=\frac{1}{N} \langle u^{\otimes N}, H_N^V u^{\otimes N}\rangle  = \int_{\R^3} |\nabla u|^2 +V|u|^2 + \frac{1}{2} |u|^2 (w_N*|u^2|).
\end{equation}
Moreover, if the Hartree energy functional $\cE^{V}_{\rm{H},N}(u)$ has a unique minimizer $u_{\rm H}$, then  the ground state $\Psi_N^V$ of $H_N^V$ condensates on $u_{\rm H}$ in the sense that 
\begin{equation} \label{eq:BEC-pdm}
\lim_{N\to \infty} \frac{1}{N}\left\langle u_{\rm H}, \gamma_{\Psi_N^V} u_{\rm H} \right\rangle =1,
\end{equation}
where $\gamma_{\Psi}:\gH\to \gH$ is the one-body density matrix of $\Psi\in \gH^N$ with kernel
\be \label{eq:def-1pdm-HN}
\gamma_{\Psi}(x,y) = N \int \Psi (x,x_2, \dots , x_N) \overline{\Psi(y,x_2,\dots,x_N)}\, \mathrm{d} x_{2}\cdots \mathrm{d} x_{N}.
\ee

The rigorous justifications for \eqref{eq:BEC-energy} and \eqref{eq:BEC-pdm} in various specific cases has been given in \cite{LieLin-63,FanSpoVer-80,BenLie-83,LieYau-87,PetRagVer-89,RagWer-89,Seiringer-11,SeiYngZag-12}. Recently, in a series of works \cite{LewNamRou-14,LewNamRou-14c,LewNamRou-14d}, Lewin, Rougerie and the first author of the present paper provided proofs in a very general setting. Note that when $\beta=1$ (the Gross-Pitaevskii regime), the Hartree functional has to be modified to capture the strong correlation between particles. This has been first done by Lieb, Seiringer and Yngvason in \cite{LieSeiYng-00,LieSei-02,LieSei-06} (see also \cite{NamRouSei-15}). 

\subsubsection*{Fluctuations around the condensation} The next order correction to the lower eigenvalues and eigenfunctions of $H_N^V$ is predicted by Bogoliubov's approximation \cite{Bogoliubov-47b}. This has been first derived rigorously by Seiringer in \cite{Seiringer-11}, and then extended in various directions in \cite{GreSei-13,LewNamSerSol-15,DerNap-13,NamSei-15}.

Bogoliubov's theory is formulated in the Fock space
$$ \cF(\gH)= \bigoplus_{n=0}^\infty \gH^n= \C \oplus \gH \oplus \gH^2 \oplus \cdots,$$
where the excited particles are effectively described by a quadratic Hamiltonian $\bH^V$ acting on the subspace $\cF(\{u_{\rm H}\}^\bot)$. In fact, $\bH^V$ is the second quantization of (half) the Hessian of the Hartree functional $\cE_{\rm H}^V(u)$ at its minimizer $u_{\rm H}$.

It was proved in \cite[Theorem 2.2]{LewNamSerSol-15} by Lewin, Serfaty, Solovej and the first author of the present paper that if $\beta=0$ and if the Hartree minimizer $u_{\rm H}$ is non-degenerate (in the sense that the Hessian of $\cE_{\rm H}^V(u)$ at $u_{\rm H}$ is bigger than a positive constant), then the ground state $\Psi_N^V$ of $H_N^V$ admits the norm approximation
\begin{align} 
\label{eq:GS-norm}
\lim_{N\to \infty} \left\| \Psi_N^V - \sum_{n=0}^N u_{\rm H}^{\otimes (N-n)} \otimes_s \psi_n   \right\|_{\gH^N} =0
\end{align}
where $\Phi^V = (\psi_n)_{n=0}^\infty \in\cF(\{u_{\rm }\}^\bot)$ is the (unique) ground state of $\bH^V$. The same is expected to be true for $\beta>0$ small. 

Note that the norm convergence \eqref{eq:GS-norm} is much more precise than the convergence of density matrices in \eqref{eq:BEC-pdm}. In fact, if $w\not\equiv 0$, then $\Phi^V$ is not the vacuum $\Omega:=1\oplus 0 \oplus 0 \cdots$, and hence $\Psi_N^V$ is {\em never} close to $u_{\rm H}^{\otimes N}$ in norm.  

\subsubsection*{Quasi-free states} The ground state $\Phi^V$ of the quadratic Hamiltonian $\bH^V$ is a quasi-free state (see \cite[Theorem A.1]{LewNamSerSol-15}). Recall that a state $\Psi$ in Fock space $\cF(\gH)$ is called a quasi-free state if it has finite particle number expectation and satisfies Wick's Theorem: 
\begin{align}
&\quad\ \langle \Psi, a^{\#}(f_{1}) a^{\#}(f_{2}) \cdots a^{\#}(f_{2n-1}) \Psi \rangle = 0,  \label{eq:Wick-1}\\[1ex]
\label{eq:Wick-2}
&\quad\ \langle \Psi, a^{\#}(f_{1}) a^{\#}(f_{2}) \cdots a^{\#}(f_{2n})  \Psi \rangle\\
\notag &= \sum_{\sigma\in P_{2n}} \prod_{j=1}^n \langle \Psi, a^{\#}(f_{\sigma(2j-1)}) a^{\#}(f_{\sigma(2j)}) \Psi \rangle 
\end{align}
for all $n$ and for all $f_1,\dots,f_n \in \gH$, where $a^{\#}$ is either the creation or annihilation operator (see  Section \ref{sec:main-result}) and $P_{2n}$ is the set of pairings,
$$
P_{2n} = \{\sigma \in S(2n)~ |~ \sigma(2j-1)<\min\{\sigma(2j),\sigma(2j+1)\}~\text{for~all}~ j\}.
$$

It is clear that if $\Psi$ is a quasi-free, then the projection $|\Psi\rangle \langle \Psi|$ is determined completely by its density matrices. Recall that for every state $\Psi$ in $\cF(\gH)$, we define the density matrices $\gamma_\Psi: \gH\to \gH$ and $\alpha_\Psi:\overline{\gH} \to {\gH}$ by
\be \label{eq:def-gamma-alpha}
\left\langle {f,{\gamma _\Psi }g} \right\rangle  = \left\langle \Psi, {{a^*}(g)a(f)} \Psi \right\rangle,\quad \left\langle {{f}, \alpha _\Psi \overline{g} } \right\rangle  = \left\langle \Psi, {a(g)a(f)} \Psi\right\rangle
\ee
for all $f,g\in \gH$. If $\Psi \in \gH^N$, then $\gamma_\Psi$ coincides with the density matrix defined in \eqref{eq:def-1pdm-HN}. In general, $\Tr(\gamma_\Psi)$
is the particle number expectation of $\Psi$. 

\subsection{Time evolution} 

Let us recall the physical picture we have in mind. Let $\Psi_N(0)$ be a ground state of $H_N^V$ (cf. \eqref{eq:HN-V}). When the external potential $V$ is turned off, $\Psi_N(0)$ is no longer a ground state of the Hamiltonian $H_N$ (cf. \eqref{eq:def-HN}) and the time evolution $\Psi_N(t)=e^{-itH_N} \Psi_N(0)$ is observed. The analysis of the behavior of $\Psi_N(t)$ when $N\to \infty$ is the main goal of our paper.

\subsubsection*{Leading order} It is a fundamental fact that the Bose-Einstein condensation is stable under the Schr\"odinger flow in the mean-field limit. To be precise, if $\Psi_N(0)$ condensates on a (one-body) state $u(0)$, in the sense of~\eqref{eq:BEC-pdm}, then the time evolution $\Psi_N(t)=e^{-itH_N}\Psi_N(0)$ condensates on the state $u(t)$ determined by the Hartree equation
\begin{equation} \label{eq:Hartree-equation}
i\partial_t u(t) =  \big(-\Delta +w_N*|u(t)|^2 -\mu_N(t)\big) u(t) 
\end{equation}
and the initial datum $u(0)$. 

Here $\mu_N(t)\in \R$ is a phase parameter which is free to choose. For the leading order, this phase plays no role as it does not alter the projection $|u(t)\rangle \langle u(t)|$. However, to simplify the second order expression discussed below, we will choose  
\be \label{eq:def-mu}
\mu_N(t):=\frac12\iint_{\R^3\times\R^3}|u(t,x)|^2w_N(x-y)|u(t,y)|^2\,\mathrm{d} x\,\mathrm{d} y
\ee
which ensures the compatibility of the energies:
\begin{align*}
\left\langle u(t)^{\otimes N}, H_N \big(u(t)^{\otimes N} \big) \right \rangle &\approx
\left\langle \Psi_{N}(t), H_N \Psi_{N}(t) \right\rangle \\
&= \left\langle \Psi_{N}(t), i \partial_t  \Psi_{N}(t)\right\rangle \approx \left\langle u(t)^{\otimes N}, i \partial_t  \big( u(t)^{\otimes N} \big) \right\rangle.
\end{align*}

Note that when $\beta>0$, $w_N \wto a_0 \delta$ weakly with $a_0 =\int w$ and the solution to the Hartree equation \eqref{eq:Hartree-equation} converges to that of the cubic nonlinear Schr\"odinger equation (NLS)
\begin{equation} \label{eq:NLS}
i\partial_t v(t) =  \big(-\Delta + a_0 |v(t)|^2 - \mu(t)\big) v(t) , \quad \mu(t)=\frac{a_0}{2}\int |v(t,x)|^4 \mathrm{d} x,
\end{equation}
with the same initial datum. 

The rigorous derivation for the dynamics of the Bose-Einstein condensation has been the subject of a vast literature. For $\beta=0$, the problem was studied by Hepp \cite{Hepp-74}, Ginibre and Velo \cite{GinVel-79,GinVel-79b} and  Spohn \cite{Spohn-80}; see \cite{BarGolMau-00,ErdYau-01,AdaGolTet-07,AmmNie-08,FroKnoSch-09,RodSch-09,KnoPic-10,Pickl-11,CheLeeSch-11}  for further results. For $0<\beta\le 1$, the problem was solved by Erd\"os, Schlein and Yau \cite{ErdSchYau-07,ErdSchYau-09,ErdSchYau-10} (see also \cite{KlaMac-08,BenOliSch-15,Pickl-15,CheHaiPavSei-15} for later developments for $\beta=1$). Note that when $\beta=1$, the strong correlation between particles yield a leading order correction to the effective dynamics and the NLS equation \eqref{eq:NLS} has to be replaced by the Gross-Pitaevskii equation, i.e. $a_0$ has to be replaced by the scattering length of $w$. 

\subsubsection*{Second order correction} In this paper, we are interested in the norm approximation for  the time evolution $\Psi_N(t)=e^{-itH_N}\Psi_N(0)$.  

For the norm approximation, a natural approach is to study the time evolution initiated by a \emph{coherent state} in Fock space $\cF(\gH)$. Recall that a coherent state is obtained by applying Weyl's unitary operator $W(f)=\exp\big(a^*(f)-a(f)\big)$ to the vacuum:
$$W(f) \Omega = e^{-\norm{f}^2/2}\sum_{n\geq0}\frac{1}{\sqrt{n!}}f^{\otimes n}.$$
The $N$-body Hamiltonian $H_N$ can be extended to the Fock space as 
\begin{equation} \label{eq:HN-Fock}
H_N = \int_{\R^3} a^*_x (-\Delta) a_x  \, \mathrm{d} x +\frac{1}{2(N-1)} \iint_{\R^3\times\R^3} w(x-y)a^*_x a^*_y a_x a_y \, \mathrm{d} x\mathrm{d} y 
\end{equation}
where $a_x^*$ and $a_x$ are the  operator-valued distributions (see Section \ref{sec:main-result}). In the mean-field regime, the time evolution of a coherent state satisfies the norm approximation 
\be \label{eq:norm-app-coherent}
\lim_{N\to \infty}\left\| \exp(-itH_N)W\left(\sqrt{N}u(0)\right)\Omega - W\left(\sqrt{N}u(t)\right) \Xi(t) \right\|_{\cF}=0
\ee
where $u(t)$ is the Hartree evolution \eqref{eq:Hartree-equation} and $\Xi(t)$ is governed by a quadratic  Hamiltonian on $\cF(\gH)$. The scale factor $\sqrt{N}$ appears naturally as the particle number expectation of the coherent state $W(f)\Omega$ is $\|f\|^2$. 

The convergence \eqref{eq:norm-app-coherent} has been justified rigorously by Hepp \cite{Hepp-74} and Ginibre and Velo \cite{GinVel-79,GinVel-79b} for $\beta=0$, and then by Grillakis and Machedon \cite{GriMac-13} for $0\le \beta<1/3$, based on their previous works with Margetis \cite{GriMacMar-10,GriMacMar-11}. 
\smallskip

Note that by projecting the Fock-space approximation \eqref{eq:norm-app-coherent} onto the $N$-particle sector $\gH^N$, it is possible to derive a norm approximation for the time evolution initiated by a Hartree state $u(0)^{\otimes N}$ (see \cite[Sec. 3]{LewNamSch-14}). This technique was first introduced by Rodnianski and Schlein in \cite{RodSch-09} to obtain the error estimate for the Hartree dynamics (see also \cite{CheLeeSch-11,BenOliSch-15}). The coherent state approach, however, has two obvious drawbacks.    

\begin{itemize}
\item [$\bullet$] First, projecting from Fock space to $\gH^N$ makes certain estimates {\em weaker}. For example, it was shown in \cite{GriMac-13} that the coherent state approximation \eqref{eq:norm-app-coherent} is valid for all $0\le \beta<1/3$, but this only gives a meaningful approximation on $\gH^N$ for $0\le \beta<1/6$. 

\item Second, and more seriously, the initial state $u(0)^{\otimes N}$ is not really the physically relevant one. Recall that the ground state of $H_N^V$ admits the approximation \eqref{eq:GS-norm} and it is never close to a Hartree state $u(0)^{\otimes N}$ in norm (except when $w\equiv 0$).
\end{itemize}

Recently, a direct approach for $N$-particle initial states has been proposed in \cite{LewNamSch-14} by Lewin, Schlein and the first author of the present paper, based on ideas introduced in \cite{LewNamSerSol-15}. They considered the $N$-particle initial states of the form
\be \label{eq:PsiN0-intro}
\Psi_N(0) = \sum_{n=0}^N u(0)^{\otimes (N-n)} \otimes_s \psi_n(0)
\ee
where $(\psi_n(0))_{n=0}^\infty \in \cF(\{u(0)\}^\bot)$. This form is motivated by the ground state property \eqref{eq:GS-norm} of trapped systems. It was proved in \cite{LewNamSch-14} that when $\beta=0$, the time evolution $\Psi_N(t)=e^{-itH_N}\Psi_N(0)$ satisfies the norm approximation
\be \label{eq:PsiN-intro}
\lim_{N\to \infty} \left\| \Psi_N(t) - \sum_{n=0}^N u(t)^{\otimes (N-n)} \otimes_s \psi_n(t) \right\|_{\gH^N} =0
\ee
where $u(t)$ is the Hartree evolution \eqref{eq:Hartree-equation} and $\Phi(t)=(\psi(t))_{n=0}^\infty \in \cF(\{u(t)\}^\bot)$ is generated by a quadratic Bogoliubov Hamiltonian (which is different from the effective Hamiltonian governing the evolution of $\Xi(t)$ in \eqref{eq:norm-app-coherent}). 
\medskip

In the present paper, we prove that the norm convergence \eqref{eq:PsiN-intro} holds true for all $0\le \beta < 1/3$. This result can not be obtained by a straightforward modification of the proof in \cite{LewNamSch-14}. The main  new ingredient is a different way to control the number of the particles outside of the condensate when $\beta>0$. More precisely, we will show that the one-particle density matrices of the Bogoliubov dynamics $\Phi(t)$ satisfy a pair of Schr\"odinger-type linear equations, which then allow us to obtain the desired bound on the number of particles in the state $\Phi(t)$ by PDE techniques. Our approach is inspired by \cite{GriMac-13} where similar equations were used. However, our derivation of the equations is different and much simpler than that of \cite{GriMac-13}.

The condition $0\le \beta<1/3$ is typical for the mean-field regime.  If $\beta > 1/3$, then the analysis becomes more complicated due to the strong correlation between particles. Very recently, independently to our work, new results for the Fock space norm approximation with $\beta>1/3$ have been obtained by Boccato, Cenatiempo and Schlein \cite{BocCenSch-15}, Grillakis and Machedon \cite{GriMac-15} and Kuz \cite{Kuz-15b}. In \cite{Kuz-15b}, the author extends the result in \cite{GriMac-13} for $\beta<1/2$. In \cite{GriMac-15}, the authors prove a result similar to \eqref{eq:norm-app-coherent} for $\beta < 2/3$, but now the mean-field dynamics $u(t)$ and the quadratic generator have to be modified. In \cite{BocCenSch-15}, the authors consider initial data of the form $W(\sqrt{N}u(0))\Phi(0)$ with a special quasi-free state $\Phi(0)$ and their result holds for $\beta <1$. 

Note that in our paper, we do not put any special assumptions on the initial states except the known properties of ground states of trapped systems. Therefore, the requirement $0<\beta<1/3$ seems to be reasonable for this large class of initial data. We expect that by using ideas from \cite{BocCenSch-15,GriMac-15}, we can improve our result for larger $\beta$ (and for more specified initial data), although the analysis in $N$-particle Hilbert space should be more complicated than that in Fock space. We hope to come back to this issue in the future. 

In the most interesting case, $\beta=1$, the norm approximation to the quantum dynamics is an open problem.

Finally, let us remark that our method is quite general and it can be applied to many different situations. For example, our result can be extended to $d=1$ or $d=2$ dimensions with attractive interaction potential (i.e. $w<0$), with or without external potential, provided that $\beta<1/d$ (cf.  Remarks~\ref{rem:dimensions} and~\ref{rem:focusing}).

The precise statement of our main result will be given in the next section.

\section{Main result} \label{sec:main-result}

In this section we present our overall strategy and state our main theorem.

\subsection{Fock space formalism} Let us quickly recall the Fock space formalism which is used throughout the paper. On the Fock space
$$ \cF(\gH)= \bigoplus_{n=0}^\infty \gH^n= \C \oplus \gH \oplus \gH^2 \oplus \cdots$$ 
we can define the creation operator $a^*(f)$ and the annihilation operator $a(f)$ for every $f\in \gH$ by
\begin{align*}
(a^* (f) \Psi )(x_1,\dots,x_{n+1})&= \frac{1}{\sqrt{n+1}} \sum_{j=1}^{n+1} f(x_j)\Psi(x_1,\dots,x_{j-1},x_{j+1},\dots, x_{n+1}) \\
 (a(f) \Psi )(x_1,\dots,x_{n-1}) &= \sqrt{n} \int \overline{f(x_n)}\Psi(x_1,\dots,x_n) \mathrm{d} x_n 
\end{align*}
for all $\Psi\in \gH^n$ and for all $n$. These operators satisfy the canonical commutation relations (CCR)
\be \label{eq:ccr}
[a(f),a(g)]=[a^*(f),a^*(g)]=0,\quad [a(f), a^* (g)]= \langle f, g \rangle
\ee
for all $f,g\in \gH$. The creation and annihilation operators are widely used to represent many other operators on Fock space. The following result is well-known; see e.g. \cite{Berezin-66} or \cite[Lemmas 7.8 and 7.12]{Solovej-ESI-2014}.

\begin{lemma}[Second quantization] \label{lem:second-quantization} Let $H$ be a symmetric operator on $\gH$ and let $\{f_n\}_{n\ge 1}\subset D(h)$ be an orthonormal basis for $\gH$. Then 
\be  \label{eq:second-quantization-H}
\dGamma (H) := 0\oplus  \bigoplus_{N=1}^\infty \sum_{j=1}^N H_j = \sum_{m,n\ge 1} \langle f_m,H f_n\rangle a^*(f_m)a(f_n).
\ee
Let $W$ be a symmetric operator on $\gH \otimes \gH$ and let $\{f_n\}_{n\ge 1}$ be an orthonormal basis for $\gH$ such that $ f_m \otimes f_n \in D(W)$ and 
$$\langle f_m\otimes f_n, W \,f_p\otimes f_q \rangle = \langle f_n\otimes f_m, W \,f_p\otimes f_q \rangle$$
for all $m,n,p,q \ge 1$. Then 
\begin{align} \label{eq:second-quantization-W}
&\quad\ 0\oplus 0 \oplus \bigoplus_{N=2}^{\infty} \sum_{1\le i<j \le N} W_{ij}\\
\nn &= \frac{1}{2} \sum_{m,n,p,q\ge 1} \langle f_m\otimes f_n, W \,f_p\otimes f_q \rangle_{\gH^2} \,\,a^*(f_m)a^*(f_n)a(f_p) a(f_q).
\end{align} 
\end{lemma} 

If one does not want to work on a specific orthonormal basis, it is possible to use the operator-valued distributions $a_x^*$ and $a_x$, with $x\in\R^3$, defined by
$$
a^*(f)=\int_{\R^3}  f(x) a_x^* \mathrm{d} x \qquad \text{and} \qquad a(f)=\int_{\R^3} \overline{f(x)} a_x \mathrm{d} x
$$
for all $f\in \gH$. The canonical commutation relations \eqref{eq:ccr} then imply that
\be \label{eq:ccr-x}
[a^*_x,a^*_y]=[a_x,a_y]=0 
\qquad \text{and}\qquad [a_x,a^*_y]=\delta(x-y).
\ee
The second quantization formulas \eqref{eq:second-quantization-H} and \eqref{eq:second-quantization-W} can be rewritten as
\begin{align}  \label{eq:second-quantization-Hxy}
&\quad\ \dGamma (H) = \iint H(x,y)a^*_x a_y \, \mathrm{d} x \mathrm{d} y, \\[1ex]
 \label{eq:second-quantization-Wxy}
&\quad\ 0\oplus 0 \oplus \bigoplus_{N=2}^{\infty} \sum_{1\le i<j \le N} W_{ij}\\
\nn &= \frac{1}{2} \iiiint W(x,y;x',y') a^*_x a^*_y a_{x'} a_{y'} \, \mathrm{d} x \mathrm{d} y \mathrm{d} x' \mathrm{d} y'
\end{align} 
where $H(x,y)$ and $W(x,y;x',y')$ are the kernels of $H$ and $W$, respectively.

For example, the particle number operator can be written as 
$$
\cN := \dGamma(1) = \bigoplus_{n=0}^n n \1_{\gH^n} = \int_{\R^3} a_x^* a_x \, \mathrm{d} x
$$
and the $N$-body Hamiltonian $H_N$ can be extended to an operator on Fock space $\cF(\gH)$ as 
\be
H_N = \dGamma(-\Delta)+\frac{1}{2(N-1)}\iint_{\R^3\times \R^3} w_N(x-y)a^*_x a^*_y a_x a_y\, \mathrm{d} x \mathrm{d} y. \label{def:2ndquantHamilt}
\ee

\subsection{Fluctuations around Hartree states} 

As discussed in the introduction, the starting point of our analysis is the Bose-Einstein condensation described by the Hartree equation. The following well-posedness of Hartree equation is taken from \cite[Proposition 3.3]{GriMac-13}.

\enlargethispage{1em}
\begin{lemma}[Hartree evolution] \label{lem:Hartree-evolution} For every initial datum $u(0,\cdot)\in H^{s}(\R^3)$, $s\ge 1$, the Hartree equation \eqref{eq:Hartree-equation} has a unique global solution $u(t,x)$ and 
$$ \|u(t,\cdot)\|_{H^s(\R^3)} \le C <\infty$$
for a constant $C$ depending only on $\|u(0,\cdot)\|_{H^s(\R^3)}$ (independent of $N$ and $\beta$). Moreover, if $u(0)\in W^{\ell,1}(\R^3)$ with $\ell$ sufficiently large, then 
$$ \|u(t)\|_{L^\infty(\R^3)}\le \frac{C_1}{(1+t)^{3/2}}$$
for a constant $C_1$ depending only on $\|u(0)\|_{W^{\ell,1}(\R^3)}$.
\end{lemma}  

A similar result for the cubic NLS has been proved in \cite{Bourgain-98}. 
In the following, we will always denote by $u(t)=u(t,.)$ the solution to the Hartree equation \eqref{eq:Hartree-equation} with an initial datum $u(0)\in H^{2}(\R^3)$. In particular, by Sobolev's embedding $H^2(\R^3)\subset C(\R^3)$ we have the uniform bound $ \|u(t)\|_{L^\infty(\R^3)} \le C$ for a constant $C$ depending only on $\|u(0)\|_{H^2(\R^3)}$.

\smallskip

To describe the particles outside of the condensate, we introduce $$Q(t):=1-|u(t) \rangle \langle u(t)|,\quad\gH_+(t):=Q(t)\gH = \{u(t)\}^\bot$$
and the excited Fock space $\cF_+(t) \subset \cF(\gH)$:
$$ \cF_+(t) := \cF(\gH_+(t)) = \bigoplus_{n=0}^\infty \gH_+(t)^n = \bigoplus_{n=0}^\infty \bigotimes^n_{\rm sym} \gH_+(t) .$$
The corresponding particle number operator is 
$$\cN_+(t):= \dGamma(Q)= \bigoplus_{n=0}^\infty n \1_{\gH_+^n(t)} = \cN-a^*(u(t))a(u(t)).$$
As in~\cite[Sec. 2.3]{LewNamSerSol-15}, we can decompose any function $\Psi\in \gH^N$ as 
\begin{equation*}
\Psi=\sum_{n=0}^N u(t)^{\otimes (N-n)} \otimes_s \psi_n  = \sum_{n=0}^N \frac{(a^*(u(t)))^{N-n}}{\sqrt{(N-n)!}} \psi_n
\label{eq:decomp-Psi-UN} 
\end{equation*}
with $\psi_n \in \gH_+(t)^{n}$, and this gives rises the {\em unitary} operator
\begin{equation}
 \label{eq:def-unitary-UNt}
\begin{aligned}
U_{N}(t):  \gH^N & \to  \displaystyle \cF_+^{\le N}(t):=\bigoplus_{n=0}^N \gH_+(t)^n \\[1ex]
  \Psi & \mapsto  \psi_0\oplus \psi_1 \oplus\cdots \oplus \psi_N.
\end{aligned}
\end{equation}



In our analysis, the unitary operator $U_N(t)$ plays the same role as Weyl's unitary operator $W(\sqrt{N}u(t))$, which has been used in \cite{Hepp-74,GinVel-79,GinVel-79b,GriMacMar-10,GriMacMar-11,GriMac-13} to investigate the fluctuations around coherent states. However, the operator $U_N(t)$ is more suitable to play with on $N$-particle sector $\gH^N$. 

\subsection{Bogoliubov's approximation} \label{sec:main-thm-Bog-app}Following \cite{LewNamSch-14}, we will consider 
\be \label{eq:def-PhiNt}
\Phi_N(t):=U_N(t) \Psi_N(t).
\ee
The vector $\Phi_N(t)$ belongs to $\cF_+^{\le N}(t)$ and it satisfies the equation
\be \label{eq:eq-PhiNt}
\begin{cases}
i \partial_t \Phi_N(t)  =  \Big[ i \left( \partial_t U_N(t)\right) U_N^*(t) + U_N(t)H_N U_N^*(t) \Big] \Phi_N(t), \\
 \Phi_N(0)  = U_N(0) \Psi_N(0).
\end{cases}
\ee
The first key ingredient in our approach is the following approximation
\be \label{eq:Bog-app-newHNt}
i \left( \partial_t U_N(t)\right) U_N^*(t) + U_N(t)H_N U_N^*(t) \approx \bH(t),
\ee
where $\bH(t)$ is derived from Bogoliubov's theory:
\begin{align} \label{eq:Bogoliubov-Hamiltonian-Ht}
\bH(t)&:= \dGamma \big(-\Delta+|u(t)|^2\ast w_N -\mu_N(t) + K_1(t) \big)\\
&+\frac12\iint_{\R^3\times\R^3}\Big(K_2(t,x,y)a^*(x)a^*(y)+\overline{K_2(t,x,y)}a(x)a(y)\Big)\mathrm{d} x\,\mathrm{d} y.\nn
\end{align}
Here $K_1(t)={Q(t)}\widetilde{K}_1(t)Q(t)$ where $\widetilde{K}_1(t)$ is the operator on $\gH$ with kernel $\widetilde{K}_1(t,x,y)\!=\!u(t,x)w_N(x-y)\overline{u(t,y)}$, and $K_2(t,\cdot,\cdot)\!=\!Q(t)\otimes Q(t)\widetilde{K}_2(t,\cdot,\cdot) \!\in\! \gH^2$ with $\widetilde{K}_2(t,x,y)=u(t,x)w(x-y)u(t,y)$. 
\smallskip

When $\beta=0$, the approximation \eqref{eq:Bog-app-newHNt} in the meaning of quadratic forms has been justified in \cite{LewNamSch-14}, inspired by ideas in \cite{LewNamSerSol-15}. To deal with the case $0\le \beta<1/3$, we will need the following operator bound.

\begin{proposition}[Bogoliubov's approximation] \label{lem:Bog-app} Let $\beta\ge 0$ and $N\in \mathbb{N}$ arbitrary. Let $u(t)$ be the Hartree evolution with initial datum $u(0)\in H^2(\R^3)$. Denote
$$ R(t)= \1_{\cF_+^{\le N}(t)} \Big[i \left( \partial_t U_N(t)\right) U_N^*(t) + U_N(t)H_N U_N^*(t) -  \bH(t) \Big] \1_{\cF_+^{\le N}(t)}.$$
Then $R(t)=R^*(t)$ and 
\begin{align} \label{eq:Bog-app-RR*}
R^2(t)\le C \Big( N^{6\beta-2} \cN_+^4(t) + N^{3\beta-1}\cN_+^3(t) + N^{3\beta-2} \Big)
\end{align}
on $\cF(\gH)$, for a constant $C$ depending only on $\|u(0)\|_{H^2(\R^3)}$.
\end{proposition}

A bound similar to \eqref{eq:Bog-app-RR*} has been used in \cite[Theorem 1]{NamSei-15} to study the collective excitation spectrum and stationary states of mean-field Bose gases. For the reader's convenience, we will provide a full proof of Proposition \ref{lem:Bog-app} in Section \ref{sec:Bog-app}. 

Recall that we are interested in the evolution of the $N$-particle initial states of the form \eqref{eq:PsiN0-intro}: 
$$
\Psi_N(0) = \sum_{n=0}^N u(0)^{\otimes (N-n)} \otimes_s \psi_n(0)
$$
where $\Phi(0):=(\psi_n(0))_{n=0}^\infty \in \cF_+(0)$. Under this choice,
$$
\Phi_N(0)=U_N \Psi_N(0) = (\psi_n(0))_{n=0}^N 
$$
converges in norm to $\Phi(0)$. Combining with Bogoliubov's approximation \eqref{eq:Bog-app-newHNt}, we may expect that the evolution $\Phi_N(t)$ in \eqref{eq:eq-PhiNt} is close (in norm) to the solution of the effective Bogoliubov equation
\be \label{eq:eq-Phit}
\begin{cases}
i \partial_t \Phi(t) =  \bH(t) \Phi(t), \\
\Phi(t=0) = \Phi(0).
\end{cases}
\ee

The existence and uniqueness of the solution of \eqref{eq:eq-Phit} in the quadratic form domain of $\bH(t)$ have been proved in \cite[Theorem 7]{LewNamSch-14}.  Moreover, the proof in \cite{LewNamSch-14} also gives a bound on $\langle \Phi(t), \cN \Phi(t)\rangle$ which, in particular, depends on $\|K_2(t,\cdot,\cdot)\|_{L^2(\R^3\times \R^3)}$. Indeed, a natural way to bound $\langle \Phi(t), \cN \Phi(t)\rangle$ is to compute the derivative 
$$
\frac{d}{dt} \langle \Phi(t), \cN \Phi(t)\rangle = - \langle \Phi(t), i[\cN, \bH] \Phi(t)\rangle
$$ 
and then use Gr\"onwall's inequality. This requires a bound on the commutator $i[\cN, \bH]$ in terms of $\cN$. To our knowledge, the best known bound of this type is  
$$
i[\cN, \bH] \le C\|K_2(t,\cdot,\cdot)\|_{L^2(\R^3\times \R^3)} (\cN+1)
$$
(see e.g. \cite[Lemma 9]{LewNamSch-14}). Unfortunately, when $\beta>0$,  
$$
\|K_2(t,\cdot,\cdot)\|_{L^2(\R^3\times \R^3)}^2 \sim \iint |u(t,x)|^2 |w_N(x-y)|^2 |u(t,y)|^2\, \mathrm{d} x \mathrm{d} y \sim  N^{3\beta},
$$ 
and the Gr\"onwall argument gives a bound on $\langle \Phi(t), \cN \Phi(t)\rangle$ of the order $\exp(N^{3\beta/2})$, which is too big for our purposes. 
\medskip

The main new ingredient in our paper is a uniform bound on\linebreak $\langle \Phi(t), \cN \Phi(t)\rangle$, for any $\beta\ge 0$. More precisely, we have the following 

\begin{proposition}[Bogoliubov equation] \label{lem:evo-quasi-free} Let $\beta\ge 0$ and $N\in \mathbb{N}$ arbitrary. Let $u(t)$ be the Hartree evolution with initial datum $u(0)\in H^2(\R^3)$. Then for every initial state $\Phi(0)\in \cF_+(0)$ satisfying $\langle \Phi(0), \cN \Phi(0)\rangle <\infty$, the equation~\eqref{eq:eq-Phit} has a unique global solution $\Phi(t)$. Moreover, $\Phi(t)\in \cF_+(t)$ for all $t\ge 0$ and the following statements hold true.
\begin{itemize}
\item[(i)] The pair of density matrices $(\gamma(t),\alpha(t))=(\gamma_{\Phi(t)}, \alpha_{\Phi(t)})$ is the unique solution to the following system of one-body linear equations
\be \label{eq:linear-Bog-dm} 
\begin{cases}
i\partial_t \gamma = h \gamma - \gamma h + K_2 \alpha - \alpha^* K_2^*, \\
i\partial_t \alpha = h \alpha + \alpha h^{\rm T} + K_2  + K_2 \gamma^{\rm T} + \gamma K_2,\\
\gamma(t=0)=\gamma_{\Phi(0)}, \quad \alpha(t=0)  = \alpha_{\Phi(0)}.
\end{cases}
\ee
Here $h(t)=-\Delta+|u(t)|^2\ast w_N -\mu_N(t) + K_1(t)$; $K_2(t):\overline{\gH}\to \gH$ is the operator with kernel $K_2(t,x,y)$; and $\gamma^{\rm T}: \overline{\gH}\to \overline{\gH}$ is the operator with kernel $\gamma^{\rm T}(t,x,y)=\gamma(t,y,x)$.
 
\item[(ii)] We have 
\begin{align}\label{eq:bound-alpha-HS-main-result}
&\|\alpha(t)\|_{\rm HS}^2+ \|\gamma(t)\|_{\rm HS}^2 \le e^{Ct} ( 1 + \|\alpha(0)\|_{\rm HS}^2 + \|\gamma(0)\|_{\rm HS}^2 ).
\end{align}

\item[(iii)] If $\Phi(0)$ is a quasi-free state, then $\Phi(t)$ is a quasi-free state for all $t$ and 
\be \label{eq:bound-N-Phit}
\langle \Phi(t), \cN \Phi(t) \rangle \le e^{Ct} \Big( 1 + \langle \Phi(0), \cN \Phi(0)\rangle \Big)^2
\ee
for a constant $C$ depending only on $\|u(0)\|_{H^2(\R^3)}$. Moreover, if $u(0)\in W^{\ell,1}(\R^3)$ with $\ell$ sufficiently large, then 
\be \label{eq:bound-N-Phit-improved}
\langle \Phi(t), \cN \Phi(t) \rangle \le C_1 \Big( \log(1+t) +  1 + \langle \Phi(0), \cN \Phi(0)\rangle \Big)^2
\ee
for a constant $C_1$ depending only on $\|u(0)\|_{W^{\ell,1}(\R^3)}$.
\end{itemize}
\end{proposition}

Proposition \ref{lem:evo-quasi-free} is a consequence of an abstract result proved in Section~\ref{sec:evo-quasi-free} on the evolution generated by a general quadratic Hamiltonian on Fock space.  

The $N$-independent estimate \eqref{eq:bound-N-Phit} plays an essential role in our analysis and it can be derived directly from the equations \eqref{eq:linear-Bog-dm}.

Our derivation of \eqref{eq:linear-Bog-dm} and \eqref{eq:bound-N-Phit} is inspired from the analysis in \cite{GriMac-13}. In fact, the bound \eqref{eq:bound-N-Phit} is similar to the paring estimate in \cite[Theorem 4.1]{GriMac-13} and the equations \eqref{eq:linear-Bog-dm} is analogous to the paring equations (17b)-(17c) in \cite{GriMac-13} (see also \cite{GriMacMar-10,GriMacMar-11} for earlier results). To be more precise, let us consider the case when $\Phi(t)$ is a quasi-free state. In this case, $\Phi(t)=T(t) \Omega$ for a unique Bogoliubov transformation on $\cF(\gH)$, and the equation \eqref{eq:eq-Phit} becomes
\be \label{eq:Bog-Tt}
\Big[ T^*(t) (i \partial_t T(t)) - T^*(t) \bH (t) T(t) \Big] \Omega =0 .
\ee
In \cite{GriMac-13}, the explicit form
$$
T(t)= \exp\left ( i\chi_N(t) + \iint \left[ \overline{k(t,x,y)} a_x a_y - k(t,x,y) a^*_x a^*_y \right] \mathrm{d} x \mathrm{d} y \right)
$$
has been taken, where $\chi_N(t) \in \R$ is a phase factor, and the pairing equations for $k(t,x,y)$ \cite[Eqs. (17b)-(17c)]{GriMac-13} have been derived such that 
$$ T^*(t) (i \partial_t T(t)) - T^*(t) \bH(t) T(t) = \dGamma(\xi)$$
for some operator $\xi:\gH\to \gH$, which ensures that \eqref{eq:Bog-Tt} holds true. 

Our derivation of the linear equations \eqref{eq:linear-Bog-dm} is different from and much shorter than the representation in \cite{GriMacMar-10,GriMacMar-11,GriMac-13}. In fact, \eqref{eq:linear-Bog-dm} follows quickly from \eqref{eq:eq-Phit} by analyzing the dynamics of the two-point correlation functions $\langle \Psi_N(t), a_x^* a_y \Psi_N(t)\rangle$ and  $\langle \Psi_N(t), a_x^* a^*_y \Psi_N(t)\rangle$.  

The first statement in (iii) is a general fact that the set of quasi-free states is stable under the evolution generated by a time-dependent quadratic Hamiltonian. This interesting statement should be well-known but we could not localize a precise reference. As pointed out to us by Jan Derezi\'nski (private communication), this statement follows from a similar statement for the evolution generated by a time-independent quadratic Hamiltonian and the closedness of the metaplectic group in Fock space. In the present paper, we will show that this statement is a direct consequence of the linear equations \eqref{eq:linear-Bog-dm}. 

The last ingredient in our approach is the following 

\begin{lemma}[Fluctuations of quasi-free states] \label{lem:Nk-quasi-free} For all $\ell\ge 1$, there exists a constant $C_\ell >0$ such that for all quasi-free states $\Psi$ in $\cF(\gH):$ 
\begin{align*}
\langle \Psi,  \cN^\ell \Psi \rangle \le  C_{\ell} (1+ \langle \Psi,  \cN \Psi \rangle )^{\ell}.
\end{align*}
\end{lemma}

This result is well-known and a proof is provided in Section \ref{sec:proof-main-thm} for completeness. In our application, the case $\ell=4$ is sufficient to control the error in Proposition \ref{lem:Bog-app}.



\subsection{Main result} Now we are able to state our main result. 

\begin{theorem}[Bogoliubov correction to mean-field dynamics] \label{thm:main} Let $u(t)$ be the Hartree evolution in \eqref{eq:Hartree-equation} with an initial state $u(0)\in H^{2}(\R^3)$. Let $\Phi(t)=(\psi_n(t))_{n=0}^\infty \in \cF_+(t)$ be the Bogoliubov evolution in \eqref{eq:eq-Phit} with an initial quasi-free state $\Phi(0)= (\psi_n(0))_{n=0}^\infty \in \cF_+(0)$. Then the Schr\"odinger evolution $\Psi_N(t)=e^{-itH_N}\Psi_N(0)$ with the initial state 
\be \label{eq:PhiN0-thm}
\Psi_N(0) = \sum_{n=0}^N u(0)^{\otimes (N-n)} \otimes_s \psi_n(0) 
\ee
satisfies the following norm approximation:
\be \label{eq:PsiN-app-thm}
\left\| \Psi_N(t) - \sum_{n=0}^N u(t)^{\otimes (N-n)} \otimes_s \psi_n(t) \right\|_{\gH^N} \le C_0(t) N^{(3\beta-1)/2}, 
\ee
where
$$
C_0(t) \le e^{Ct}  \Big( 1 + \langle \Phi(0), \cN \Phi(0) \rangle \Big)^4
$$ 
for a constant $C$ depending only on $\|u(0)\|_{H^2(\R^3)}$. Moreover, if $u(0)\in\linebreak W^{\ell,1}(\R^3)$ with $\ell$ sufficiently large, then 
$$
C_0(t) \le C_1 (1+t)\Big( 1+ \log(1+t) + \langle \Phi(0), \cN \Phi(0)\rangle \Big)^4
$$ 
for a constant $C_1$ depending only on $\|u(0)\|_{W^{\ell,1}(\R^3)}$. 
\end{theorem}

The proof of Theorem \ref{thm:main} will be provided in Section \ref{sec:proof-main-thm}. Let us give some remarks on the result.

\begin{remark}
Since $e^{-itH_N}$ is a unitary operator on $\gH^N$, the convergence 
\be \label{eq:PsiNt-app-rmk}
\lim_{N\to \infty}\left\| \Psi_N(t) - \sum_{n=0}^N u(t)^{\otimes (N-n)} \otimes_s \psi_n(t) \right\|_{\gH^N} =0
\ee
still holds when \eqref{eq:PhiN0-thm} is replaced by the weaker assumption 
$$
\lim_{N\to \infty} \left\| \Psi_N(0) - \sum_{n=0}^N u(0)^{\otimes (N-n)} \otimes_s \psi_n(0) \right\|_{\gH^N} =0.
$$
Strictly speaking, the initial vector $\Phi_N(0)$ chosen in \eqref{eq:PhiN0-thm} is not normalized, but its norm converges to 1 and the renormalization is trivial.  
\smallskip

We also note that the initial data $u(0)$ and $\Phi(0)$ in Theorem \ref{thm:main} can be chosen to be $N$-dependent, provided that the $N$-dependences of $\|u(0)\|_{H^s(\R^3)}$ and $\langle \Phi(0), \cN \Phi(0)\rangle$ can be compensated by $N^{3\beta-1}$.  
\end{remark}

\begin{remark} In many physical applications, one is often interested in the projection $|\Psi\rangle \langle \Psi|$ of a wave function instead of the wave function $\Psi$ itself. From \eqref{eq:PsiNt-app-rmk} we obtain 
\be \label{eq:cv-dm-gHN}
\lim_{N\to \infty} \Tr_{\gH^N} \Big| |\Psi_N(t)\rangle \langle \Psi_N(t)| - U_N^* \left| \Phi(t) \right\rangle \left\langle \Phi(t) \right| U_N(t) \Big| = 0.
\ee
When $\Phi(t)$ is a quasi-free state, the projection $\left| \Phi(t) \right\rangle \left\langle \Phi(t) \right|$ is determined uniquely by its density matrices. Thus $|\Psi_N (t)\rangle \langle \Psi_N(t)|$ can be well approximated in trace norm using $(u(t),\gamma(t),\alpha(t))$ which, in principle, can be computed as accurate as we want using the one-body equations \eqref{eq:eq-Phit} and \eqref{eq:linear-Bog-dm}.

Moreover, since the one-particle density matrices can be obtained by taking the partial trace, namely$$N^{-1}\gamma_{\Psi_N(t)}= \Tr_{2\to N}|\Psi_N(t)\rangle \langle \Psi_N(t)|,$$
the convergence \eqref{eq:cv-dm-gHN} implies immediately the Bose-Einstein condensation
\be \label{eq:1pdmconv}
\lim_{N\to \infty}\Tr \left| N^{-1}\gamma_{\Psi_N(t)} - |u(t)\rangle \langle u(t)| \right| =0.
\ee
Note that, when $\beta>0$, the Hartree dynamics $u(t)$ converges to the NLS dynamics $v(t)$ in \eqref{eq:NLS} as $N\to \infty$. Therefore \eqref{eq:1pdmconv} is equivalent to 
\be \nn
\lim_{N\to \infty}\Tr \left| N^{-1}\gamma_{\Psi_N(t)} - |v(t)\rangle \langle v(t)| \right| =0.
\ee
\end{remark}
\begin{remark}\label{rem:dimensions} Our result is stated and proved in three dimensions, but it can be extended straightforwardly to one and two dimensions. More precisely, in $d \le 3$ dimensions, we can consider the two-body interaction potential  of the form $w_N(x-y)= N^{d\beta} w(N^\beta(x-y))$, and the result in Theorem \eqref{thm:main} still holds true (on the right side of \eqref{eq:PsiN-app-thm} the error now becomes $C_0(t)N^{(d\beta-1)/2}$).  
\end{remark}

\begin{remark}\label{rem:focusing} Note that Lemma \ref{lem:Hartree-evolution} is the only place where we need the assumption $w\ge 0$. The rest of our proof does not require the sign assumption on $w$ (cf. Remark \ref{re:w<0}). In particular, our result can be extended to one or two dimensions with attractive interaction potential (i.e. $w<0$), provided the well-posedness of the corresponding Hartree equation, as in Lemma \ref{lem:Hartree-evolution}, still holds. In particular, our method covers the derivation of the 1D and 2D focusing dynamics with a harmonic trap (see \cite{CheHol-13, CheHol-15} for results on the leading order).
\end{remark}

\section{Bogoliubov's approximation} \label{sec:Bog-app}

In this section we justify Bogoliubov's approximation \eqref{eq:Bog-app-newHNt}. 

\begin{proof}[Proof of Proposition \ref{lem:Bog-app}] Let us denote $\1_+^{\le N}=\1_{\cF_+^{\le N}(t)}=\1(\cN_+(t)\le N)$ for short. Recall that from the calculations in \cite[Eqs. (40)-(41)]{LewNamSch-14}, we have
\begin{align}
\label{def:R_N}
R(t) &=  \1_+^{\le N} \Big[i \left( \partial_t U_N(t)\right) U_N^*(t) + U_N(t)H_N U_N^*(t) -  \bH(t) \Big] \1_+^{\le N}\\
\nn &= \frac{1}{2}\sum_{j=1}^{5} \1_+^{\le N} ( R_{j} + R_j^*) \1_+^{\le N} 
\end{align}
where
\begin{align*}
R_{1}&=R_1^*= \mathrm{d}\Gamma(Q(t)[w_N*|u(t)|^2+K_1-\mu_N(t)]Q(t))\frac{1-\cN_+(t)}{N-1},\\
R_{2}&=-2\frac{\cN_+(t)\sqrt{N-\cN_+(t)}}{N-1} a(Q(t)[w_N*|u(t)|^2]u(t)),\\
R_{3}&= \iint  K_2(t,x,y) a^*_x a^*_y \mathrm{d} x \mathrm{d} y \left(\frac{\sqrt{(N-\cN_+(t))(N-\cN_+(t)-1)}}{N-1}-1\right),\\
R_{4}&=R_4^*= \frac{1}{2(N-1)}\iiiint({Q(t)}\otimes{Q(t)}w_N Q(t)\otimes Q(t))(x,y;x',y') \\
&  \qquad \qquad \qquad \qquad \qquad \quad  \times a^*_x a^*_y a_{x'} a_{y'} \,\mathrm{d} x \mathrm{d} y \mathrm{d} x' \mathrm{d} y',\\
R_{5}& =\frac{\sqrt{N-\cN_+(t)}}{N-1}\iiiint( 1\otimes Q(t) w_N Q(t)\otimes Q(t))(x,y;x',y') \\
& \qquad \qquad \qquad \qquad \qquad   \times \overline{u(t,x)} a^*_y a_{x'} a_{y'} \,\mathrm{d} x \mathrm{d} y \mathrm{d} x' \mathrm{d} y'.
\end{align*}
By the Cauchy-Schwarz inequality we have that 
\be \label{eq:RN-R1-R5}
R^2(t) \leq 100 \sum_{j=1}^{5} \1_+^{\le N} ( R_{j} \1_+^{\le N} R_j^* + R_j^* \1_+^{\le N} R_j) \1_+^{\le N}.
\ee
Now we estimate all terms on the right side of \eqref{eq:RN-R1-R5}. We will always denote by $C$ a constant depending only on $\|u(0)\|_{\rm H^s(\R^3)}$. 
\medskip 

\noindent $\bf {j=1}$. Using $\|w_N\|_{L^1}=\|w\|_{L^1}$ we get 
\begin{align} \label{eq:bound-wN*uu}
\left\| w_N*|u(t)|^2 \right\|_{L^\infty(\R^3)} \le \|w\|_{L^1} \|u(t)\|_{L^\infty}^2 \le C.
\end{align}
Similarly, 
\begin{align} \label{eq:bound-muN} |\mu_N(t)| &= \frac12 \left| \iint_{\R^3 \times \R^3}|u(t,x)|^2 w_N(x-y)|u(t,y)|^2\,\mathrm{d} x\,\mathrm{d} y \right| \\
&\le \frac{1}{2}\|u(t)\|_{L^2(\R^3)}^2 \|u(t)\|_{L^\infty (\R^3)}^2  \|w_N \|_{L^1(\R^3)} \le C.\nn
\end{align}
Moreover, 
\begin{align*}
\left| \langle f, \widetilde K_1(t) g \rangle  \right| &= \left| \iint \overline{f(x)} u(t,x) w_N(x-y) \overline{u(t,y)} g(y) \,\mathrm{d} x\,\mathrm{d} y \right| \\
& \le \|u(t)\|^2_{L^\infty(\R^3)} \left(\iint |f(x)|^2 |w_N(x-y)| \mathrm{d} x \mathrm{d} y \right)^{1/2}  \\
&\quad \times \left(\iint |g(y)|^2 |w_N(x-y)| \mathrm{d} x \mathrm{d} y \right)^{1/2} \\
& \le C \|u(t)\|^2_{L^\infty(\R^3)}   \| f\|_{L^2(\R^3)} \| g\|_{L^2(\R^3)} 
\end{align*}
for all $f,g\in L^2(\R^3)$. Therefore,  
\be \label{eq:bound-K1}
\| K_1(t) \| = \|Q(t) \widetilde  K_1(t) Q(t)\| \le \|\widetilde  K_1(t)\| \le C \|u(t)\|^2_{L^\infty(\R^3)} \le C.
\ee
Thus, in summary,
$$
 \pm \dGamma(Q(t)[w_N*|u(t)|^2+K_1-\mu(t)]Q(t)) \le C \dGamma (Q(t))= C \cN_+(t).
$$
Since $\dGamma(Q(t)[w_N*|u(t)|^2+K_1-\mu(t)]Q(t))$ commutes with $\cN_+(t)$, we find that
\begin{align*}
R_1^2 &= \dGamma(Q(t)[w_N*|u(t)|^2+K_1-\mu(t)]Q(t))^2 \frac{(\cN_+(t)-1)^2}{(N-1)^2} \\
&\le C \cN_+(t)^2  \frac{(\cN_+(t)-1)^2}{(N-1)^2} \le C \frac{\cN_+^4(t)}{N^2}.
\end{align*}
Consequently, 
\be \label{eq:j=1}
R_1 \1_+^{\le N} R_1 \le R_1^2 \le  C \frac{\cN_+^4(t)}{N^2}.
\ee
\medskip 

\noindent $\bf {j=2}$. Note that $v:=Q(t)[w_N*|u(t)|^2]u(t)$ satisfies
\begin{align*}
\| v \|_{L^2(\R^3)} &\le \| w_N*|u(t)|^2]u(t) \|_{L^2(\R^3)} \\
&\le \|w_N*|u(t)|^2 \|_{L^\infty(\R^3)} \|u(t)\|_{L^2(\R^3)} \le C.
\end{align*}
Using $a(v)a^*(v)= \| v\|^2 + a^*(v)a(v) \le \|v\|^2 ( \cN_+(t)+1)$, we get
\begin{align*}
 R_2 \1_+^{\le N} R_2^* &= 4 \frac{\cN_+(t)\sqrt{N-\cN_+(t)}}{N-1} a(v)\1_+^{\le N} a^*(v) \frac{\cN_+(t)\sqrt{N-\cN_+(t)}}{N-1} \\
& = 4 \frac{\cN_+(t)\sqrt{N-\cN_+(t)}}{N-1} \1^{\le N-1} a(v) a^*(v)  \1^{\le N-1} \frac{\cN_+(t)\sqrt{N-\cN_+(t)}}{N-1} \\
& \le C \frac{\cN_+(t)\sqrt{N-\cN_+(t)}}{N-1} \1^{\le N-1} (\cN_+ +1) \1^{\le N-1} \frac{\cN_+(t)\sqrt{N-\cN_+(t)}}{N-1} \\
& \le C \frac{\cN_+^3}{N}.
\end{align*}
Similarly, using
\begin{align*}
R_2^* &= -2 a^*(v) \frac{\cN_+(t)\sqrt{N-\cN_+(t)}}{N-1}  \\
&= -2   \frac{(\cN_+(t)- 1)\sqrt{N-\cN_+(t)+1}}{N-1} a^*(v)
\end{align*}
we find that
\begin{align*}
 R_2 \1_+^{\le N}R_2^* &= 4 \frac{(\cN_+(t)-1)\sqrt{N-\cN_+(t)+1}}{N-1}  a^*(v) \1_+^{\le N} a(v)  \\
&\quad \times \frac{(\cN_+(t)-1)\sqrt{N-\cN_+(t)+1}}{N-1}  \\
&= 4 \frac{(\cN_+(t)-1)\sqrt{N-\cN_+(t)+1}}{N-1} \1^{\le N+1} a^*(v)  a(v)  \\
&\quad \times \1^{\le N+1}  \frac{(\cN_+(t)-1)\sqrt{N-\cN_+(t)+1}}{N-1}  \\
&\le  C \frac{(\cN_+(t)-1)\sqrt{N-\cN_+(t)+1}}{N-1} \1^{\le N+1}  (\cN_+(t)+1)  \\
&\quad \times \1^{\le N+1} \frac{(\cN_+(t)-1)\sqrt{N-\cN_+(t)+1}}{N-1} \\
& \le C \frac{\cN_+^3(t)}{N}.
\end{align*}
Thus
\begin{align} \label{eq:j=2}
R_2^* \1_+^{\le N} R_2 + R_2 \1_+^{\le N} R_2^* \le C \frac{\cN_+^3(t)}{N}. 
\end{align}

\noindent $\bf {j=3}$. We can write $R_3 = \mathbb{K_{\rm cr}} g(\N_+(t))$ where 
\begin{align} \label{eq:def-Kcr}
\mathbb{K_{\rm cr}}:= \iint K_2(t,x,y) a^*_x a^*_y\, \mathrm{d} x \, \mathrm{d} y 
\end{align}
and
$$
g(\N_+(t)):= \frac{\sqrt{(N-\N_+(t) )(N-\N_+(t) -1)}}{N-1}-1 .
$$
Let us show that 
\begin{align}
\mathbb{K_{\rm cr}}  \mathbb{K^*_{\rm cr}} + \mathbb{K^*_{\rm cr}}  \mathbb{K_{\rm cr}} &\le 2 \|K_2(t, \cdot, \cdot)\|^2_{L^2} (\cN_+(t)+1)^2 \le  CN^{3\beta }(\cN_+(t)+1)^2 .\label{eq:Kcr*-Kcr} 
\end{align}
Here we have used 
\begin{align} \label{eq:K2-L2}
\|K_2(t,\cdot, \cdot )\|_{L^2}^2 &= \|Q(t)\otimes Q(t) \widetilde K_2(t,\cdot, \cdot )\|_{L^2}^2  \le \|\widetilde K_2(t,\cdot, \cdot )\|_{L^2}^2  \\
&= \iint |u(t,x)|^2 |w_N(x-y)|^2 |u(t,y)|^2 \, \mathrm{d} x \, \mathrm{d} y \nn\\
&\le \|u(t)\|_{L^\infty(\R^3)}^2 \|w_N\|_{L^2(\R^3)}^2 \|u(t)\|_{L^2(\R^3)}^2 \le C N^{3\beta}.\nn
\end{align}
In fact, \eqref{eq:Kcr*-Kcr} is well-known (see e.g. \cite[eq. (23) and (26)]{NamSei-15}), but we offer an alternative proof below because the proof strategy will be used later to control $R_5$. First, using the Cauchy-Schwarz inequality $XY+Y^*X^*\le XX^*+Y^*Y$ we get
\begin{align} \label{eq:KcrKcr*}
&\quad\ \mathbb{K}_{\rm cr} \mathbb{K}_{\rm cr}^* = \iiiint K_2(t,x,y) \overline{K_2(t,x',y')} a^*_x a^*_y a_{x'} a_{y'} \, \mathrm{d} x \mathrm{d} y \mathrm{d} x' \mathrm{d} y' \\
& = \frac{1}{2} \iiiint \left(  K_2(t,x,y) \overline{K_2(t,x',y')} a^*_x a^*_y a_{x'} a_{y'} + {\rm h.c.} \right) \, \mathrm{d} x \mathrm{d} y  \mathrm{d} x' \mathrm{d} y' \nn\\
&\le \frac{1}{2} \iiiint \Big(  |K_2(t,x',y')|^2 a^*_x a^*_y a_x a_y \nn\\
& \hspace{6em} + |K_2(t,x,y)|^2  a^*_{x'} a^*_{y'} a_{x'} a_{y'} \Big) \, \mathrm{d} x \mathrm{d} y  \mathrm{d} x' \mathrm{d} y'\nn\\
& = \|K_2 (t,\cdot, \cdot)\|_{L^2(\R^3\times \R^3)} ^2 \cN_+(t) (\cN_+(t)-1) .\nn
\end{align}
Here we have denoted $X+{\rm h.c.}=X+X^*$ for short (h.c. stands for Hermitian conjugate). Moreover,
\begin{align*}
\mathbb{K}_{\rm cr}^* \mathbb{K}_{\rm cr} &= \iiiint \overline{K_2(t,x,y)} {K_2(t,x',y')} a_x a_y a^*_{x'} a^*_{y'} \, \mathrm{d} x \mathrm{d} y \mathrm{d} x' \mathrm{d} y' \\
& = \iiiint \overline{K_2(t,x,y)} {K_2(t,x',y')} a^*_{x'} a^*_{y'} a_x a_y  \, \mathrm{d} x \mathrm{d} y \mathrm{d} x' \mathrm{d} y' \\
& \quad +  \iiiint \overline{K_2(t,x,y)} {K_2(t,x',y')} [a_x a_y,a^*_{x'} a^*_{y'}] \, \mathrm{d} x \mathrm{d} y \mathrm{d} x' \mathrm{d} y' 
.
\end{align*}
The first term of the right side of nothing but $\mathbb{K}_{\rm cr} \mathbb{K}_{\rm cr}^*$ which has been already estimated. For the second term, using $K_2(t,x,y)=K_2(t,y,x)$ and
\begin{align*}
[a_x a_y,a^*_{x'} a^*_{y'}]  &= \delta(x'-y) a_{y'}^* a_x + \delta(x-x') a_{y'}^* a_y + \delta (y-y') a_{x'}^* a_{x}   \\
&\quad +\delta (x-y') a^*_{x'} a_y +\delta (x'-y) \delta (x-y') + \delta(x-x')\delta(y-y')
\end{align*} 
we find that
\begin{align*}
 &\quad\ \iiiint \overline{K_2(t,x,y)} {K_2(t,x',y')} [a_x a_y,a^*_{x'} a^*_{y'}] \, \mathrm{d} x \mathrm{d} y \mathrm{d} x' \mathrm{d} y' \\
 &= 4 \iiint  \overline{K_2(t,x,y)} K_2(t,y,y') a_{y'}^* a_x \mathrm{d} x \mathrm{d} y \mathrm{d} y' + 2 \iint |K_2(t,x,y)|^2 \mathrm{d} x \mathrm{d} y \\
 & = 4 \dGamma( K_2(t) K_2^*(t) )  + 2\| K_2(t, \cdot, \cdot)\|^2_{L^2(\R^3\times \R^3)}, 
\end{align*} 
and hence
\begin{align} \label{eq:Kcr*Kcr}
\mathbb{K}_{\rm cr}^* \mathbb{K}_{\rm cr} &=  \mathbb{K}_{\rm cr} \mathbb{K}_{\rm cr}^* + 4 \dGamma( K_2(t) K_2^*(t) )  + 2\| K_2(t, \cdot, \cdot)\|^2_{L^2(\R^3\times \R^3)}.
\end{align}
Here we have denoted by $K_2(t):\overline{\gH}\to \gH$ the operator with kernel $K_2(t,x,y)$. Putting differently, $K_2(t)=Q(t)\widetilde K_2(t) \overline{Q(t)}$ with $\widetilde K_2(t): \overline{\gH}\to \gH$ the operator with kernel $u(x)w_N(x-y)u(y)$. Similarly to \eqref{eq:bound-K1} we have
\be \label{eq:bound-K2-norm}
\|K_2(t)\|= \| Q(t)\widetilde K_2(t) \overline{Q(t)}\| \le \|\widetilde K_2(t)\| \le C\|u(t)\|_{L^\infty(\R^3)}^2 \le C.
\ee
Therefore, $\dGamma( K_2(t) K_2^*(t) ) \le C \cN_+(t)$. Thus \eqref{eq:Kcr*-Kcr} follows \eqref{eq:Kcr*Kcr} and \eqref{eq:KcrKcr*}.  

Now from $R_3= \mathbb{K_{\rm cr}} g(\cN_+(t))$, using \eqref{eq:Kcr*-Kcr}, \eqref{eq:K2-L2} and the simple estimates 
$$ \1_+^{\le N-2} g^2(\cN_+(t)) + \1_+^{\le N+2} g^2(\cN_+(t)-2)  \le 32 \frac{(\cN_+(t)+1)^2}{N^2} \1_+^{\le N}$$
we conclude that 
\begin{align} \label{eq:j=3}
&\quad\ R_3^* \1_+^{\le N} R_3 + R_3 \1_+^{\le N} R_3^* \\
&= g(\cN_+(t)) \mathbb{K_{\rm cr}}^*  \1_+^{\le N} \mathbb{K_{\rm cr}} g(\cN_+(t)) + \mathbb{K_{\rm cr}} g(\cN_+(t)) \1_+^{\le N}  g(\cN_+(t)) \mathbb{K^*_{\rm cr}} \nn \\
& = g^2(\cN_+(t)) \1_+^{\le N-2} \mathbb{K_{\rm cr}}^*\mathbb{K_{\rm cr}} + g^2(\cN_+(t)-2)\1_+^{\le N+2} \mathbb{K_{\rm cr}} \mathbb{K^*_{\rm cr}} \nn\\
& \le \left( g^2(\cN_+(t)) \1_+^{\le N-2}  +  g^2(\cN_+(t)-2)\1_+^{\le N+2} \right) \Big( \mathbb{K_{\rm cr}}^*\mathbb{K_{\rm cr}} + \mathbb{K_{\rm cr}} \mathbb{K^*_{\rm cr}} \Big) \nn\\
& \le C N^{3\beta-2} (\cN_+(t)+1)^4.\nn
\end{align}
Here we have also used the fact that  $\mathbb{K_{\rm cr}}^*\mathbb{K_{\rm cr}}$ and $\mathbb{K_{\rm cr}} \mathbb{K^*_{\rm cr}}$ commute with $\cN_+(t)$. 
\medskip 

\noindent $\bf {j=4}$. By \eqref{eq:second-quantization-Wxy}, for every two-body operator $W\ge 0$ one has
\begin{align} \label{eq:second-quantization-W>0}
& \frac{1}{2} \iiiint W(x,y;x',y') 
a^*_x a^*_y a_{x'} a_{y'}\,\mathrm{d} x \mathrm{d} y \mathrm{d} x' \mathrm{d} y'
 \ge 0
\end{align}
where $W(x,y;x',y')$ is the kernel of $W$. Consequently, 
\begin{align*}
\pm R_{4} &= \pm \frac{1}{2(N-1)}\iiiint({Q(t)}\otimes{Q(t)}w_N Q(t)\otimes Q(t))(x,y;x',y') \\
&  \qquad \qquad \qquad \qquad \quad   \times  a^*_x a^*_y a_{x'} a_{y'}\,\mathrm{d} x \mathrm{d} y \mathrm{d} x' \mathrm{d} y', \\
&\le \frac{\|w_N\|_{L^\infty(\R^3)}}{2(N-1)}\iiiint( Q(t)\otimes Q(t))(x,y;x',y')a^*_x a^*_y a_{x'} a_{y'}\,\mathrm{d} x \mathrm{d} y \mathrm{d} x' \mathrm{d} y' \\
& \le C N^{3\beta-1} \cN_+^2(t). 
\end{align*}
Here we have used the simple bound $\|w_N \|_{L^\infty(\R^3)}= N^{3\beta} \|w\|_{L^\infty(\R^3)}$ in the last estimate. Since $R_4$ commutes with $\cN_+$ we find that
\be \label{eq:j=4}
R_4 \1_+^{\le N} R_4 \le R_4^2 \le C N^{6\beta-2} \cN_+^4(t).
\ee

\noindent $\bf {j=5}$. This is the most complicated case. Recall that 
$$
R_5 = \frac{\sqrt{N-\cN_+(t)}}{N-1} R_6 
$$
with
\begin{align*}
R_6 := \iiiint (1\otimes Q(t) w_N Q(t)\otimes Q(t))(x,y;x',y') \overline{u(t,x)} a_y^* a_{x'}a_{y'}\, \mathrm{d} x \mathrm{d} y \mathrm{d} x' \mathrm{d} y'.
\end{align*}
We will show that
\be \label{eq:R6*R6+R6R6*}
R_6^* R_6 + R_6 R_6^* \le C N^{3\beta} \cN^3_+(t).
\ee
\begin{remark} \label{re:w<0} Note that in the following we use $w\ge 0$, but the proof can be adapted easily to cover any $w$ without the sign assumption by decomposing $w=w_+-w_-$ and treating each term $w_\pm$ separately.
\end{remark}
 
We will write $Q=Q(t)$ and $u=u(t)$ for short. We denote
$$Q(x,y)= \delta(x-y)- u(x) \overline{u(y)},$$
the kernel of $Q$, and introduce the operators
\be 
\label{eq:def-bx-Bx} b_x := \int_{\R^3} Q(x,y) a_y \mathrm{d} y, \quad B_x := \int_{\R^3} w_N(x-y) b_y^* b_y \mathrm{d} y \ge 0.
\ee
The advantage of these notations is that using  
\begin{align*} (1\otimes Q w_N Q\otimes Q)(x,y;x',y') = \int Q(y,y_1) w_N(x-y_1) Q(x,x') Q(y_1,y') \, \mathrm{d} y_1
\end{align*}
we can rewrite
\begin{align} \label{eq:R6-reform}
R_6 &= \iint\!\!\!\iiint    Q(y,y_1) w_N(x-y_1)\\
&\hspace{5em}\times Q(x,x') Q(y_1,y') \overline{u(x)} a^*_y a_{x'} a_{y'} \,\mathrm{d} x \mathrm{d} y \mathrm{d} x' \mathrm{d} y' \nn\\
&= \iint \overline{u(x)} w_N(x-y_1) b^*_{y_1} b_{y_1} b_x \, \mathrm{d} x  \mathrm{d} y_1 = \int \overline{u(x)} B_x b_x. \nn
\end{align}
Let us list some basic properties of $b_x$ and $B_x$ defined in \eqref{eq:def-bx-Bx}. From the CCR~\eqref{eq:ccr-x} it is straightforward to see that
\be \label{eq:bx-prop}
[b_x,b_y]=0=[b^*_x, b^*_y], \quad [b_x,b_y^*] = Q(x,y) = \delta(x-y)- u(x) \overline{u(y)}.
\ee
Moreover,
\begin{align} \label{eq:int-bx}
\int b_x^* b_x \, \mathrm{d} x &= \iiint Q(z,x)  Q(x,y) a^*_z a_y \, \mathrm{d} x \mathrm{d} y \mathrm{d} z \\
&= \iint Q(z,y) a^*_z a_y \, \mathrm{d} y \mathrm{d} z = \dGamma(Q)= \cN_+(t) \nn
\end{align}
and consequently,
\begin{align} B_x &\le \|w_N\|_{L^\infty} \int b_y^* b_y \mathrm{d} y \le C N^{3\beta} \cN_+(t), \label{eq:Bx<=} \\
\int B_x \mathrm{d} x & = \int \left( \int w_N (x-y) dx \right) b_y^* b_y \mathrm{d} y \le C \cN_+(t), \label{eq:int-Bx<=} \\
\int B_x^2 \mathrm{d} x & \le C N^{3\beta} \cN_+(t) \int B_x \mathrm{d} x \le C N^{3\beta} \cN^2_+(t). \label{eq:int-BxBx<=}
\end{align}
In the last estimate we have used the fact that $B_x$ commutes with $\cN_+(t)$.
 
Now using \eqref{eq:R6-reform} we can write
\begin{align} \label{eq:R6-R6*-decomp}
R_6 R_6^* &= \iint \overline{u(x)} {u(y)} B_x b_x b_y^* B_y \, \mathrm{d} x \mathrm{d} y  \\
&= \iint \overline{u(x)} {u(y)} B_x  b_y^* b_x B_y \, \mathrm{d} x \mathrm{d} y\nn\\
&\quad + \iint \overline{u(x)} {u(y)} B_x [b_x, b_y^*] B_y \, \mathrm{d} x \mathrm{d} y .\nn
\end{align}
The second term of \eqref{eq:R6-R6*-decomp} can be estimated easily using \eqref{eq:bx-prop} and \eqref{eq:int-BxBx<=}:
\begin{align} \label{eq:R6-R6*-ineq1}
&\quad\ \iint \overline{u(x)} {u(y)} B_x [b_x, b_y^*] B_y \mathrm{d} x \mathrm{d} y \\
&= \int |u(x)|^2 B_x^2 \mathrm{d} x -  \left( \int |u(x)|^2 B_x \mathrm{d} x \right)^2 \nn\\
& \le \|u\|_{L^\infty(\R^3)}^2 \int B_x^2 \mathrm{d} x  \le C N^{3\beta} \cN_+^2(t).\nn  
\end{align}
To estimate the first term of \eqref{eq:R6-R6*-decomp}, we employ the Cauchy-Schwarz inequality $XY+Y^*X^*\le XX^*+Y^*Y$ and obtain
\begin{align} \label{eq:R6-R6*-ineq2}
&\quad\ \iint \overline{u(x)} {u(y)} B_x b_y^* b_x B_y \, \mathrm{d} x \mathrm{d} y \\
& = \frac{1}{2} \iint \left( \overline{u(x)} {u(y)} B_x  b_y^* b_x B_y + {\rm h.c.} \right) \, \mathrm{d} x \mathrm{d} y \nn\\
& \le  \frac{1}{2} \iint  \left(|u(x)|^2 B_x b_y^* b_y B_x + |u(y)|^2 B_y b_x^* b_x B_y \right) \, \mathrm{d} x \mathrm{d} y \nn\\
& \le \|u\|^2_{L^\infty(\R^3)} \iint B_x b_y^* b_y B_x \, \mathrm{d} x \mathrm{d} y \le C N^{3\beta}\cN_+^3(t). \nn
\end{align}
Here the last estimate follows from \eqref{eq:int-bx}, \eqref{eq:int-BxBx<=} and the fact that $B_x$ commutes with $\cN_+$. Thus, in summary, from \eqref{eq:R6-R6*-decomp}-\eqref{eq:R6-R6*-ineq1}-\eqref{eq:R6-R6*-ineq2} we get
\begin{align} \label{eq:R6R6*-final}
R_6 R_6^* \le C N^{3\beta}\cN_+^3(t).
\end{align} 
Now we  estimate 
\begin{align} \label{eq:R6*R6}
R_6^* R_6 &= \iint {u(x)} \overline {u(y)} b_x^* B_x B_y b_y \, \mathrm{d} x \mathrm{d} y \\
&= \iint {u(x)} \overline {u(y)} b_x^* B_y B_x b_y \, \mathrm{d} x \mathrm{d} y \nn\\
&\quad + \iint {u(x)} \overline {u(y)} b_x^* [B_x,B_y] b_y \, \mathrm{d} x \mathrm{d} y.\nn
\end{align}
The first term of \eqref{eq:R6*R6} can be bounded similarly to the first term of \eqref{eq:R6-R6*-decomp}. Indeed, by the Cauchy-Schwarz inequality $XY+Y^*X^*\le XX^*+Y^*Y$ and  \eqref{eq:int-bx}, \eqref{eq:int-BxBx<=}, we have
\begin{align} \label{eq:R6*R6-ineq1}
&\quad\ \iint {u(x)} \overline{u(y)} b_x^* B_y B_x b_y \, \mathrm{d} x \mathrm{d} y \\
& =\frac{1}{2}\iint \Big( {u(x)} \overline{u(y)} b_x^* B_y B_x b_y + {\rm h.c.} \Big)  \, \mathrm{d} x \mathrm{d} y \nn\\
& \le  \frac{1}{2}\iint \Big( |u(x)|^2 b_x^* B_y^2 b_x + |u(y)|^2 b_y^* B_x^2 b_y \Big) \, \mathrm{d} x \mathrm{d} y \nn\\
& \le \|u\|_{L^\infty(\R^3)}^2 \iint b_x^* B_y^2 b_x \mathrm{d} x \mathrm{d} y \nn\\
& \le C N^{3\beta}  \int b_x^* \cN^2_+(t) b_x \mathrm{d} x \nn\\
& =  C N^{3\beta}  \int  b_x^*  b_x (\cN_+(t)+ 1)^2  \mathrm{d} x\nn\\
&\le C N^{3\beta} \cN^3_+(t).\nn
\end{align}
To estimate the second term of \eqref{eq:R6*R6}, we use 
\begin{align*}
\left[b_{x'}^*b_{x'}, b_{y'}^* b_{y'}  \right] &= b_{x'}^* [b_{x'},b_{y'}^*] b_{y'} - b^*_{y'}[b_{y'},b_{x'}^*] b_{x'}  \\
& = Q(x',y')  b_{x'}^*  b_{y'} - Q(y',x') b^*_{y'} b_{x'}   \\
&= - u(x')\overline{u(y')} b_{x'}^*  b_{y'}  + u(y')\overline{u(x')} b^*_{y'} b_{x'}
\end{align*}
and write
\begin{align*}
&\quad\ \iint  {u(x)} \overline{u(y)} b_x^* [B_x,B_y]b_y \mathrm{d} x \mathrm{d} y \\
& = \iint  {u(x)} \overline{u(y)} b_x^* \left[\int w_N(x-x') b_{x'}^*b_{x'} \mathrm{d} x' , \int w_N(y-y') b_{y'}^* b_{y'} \mathrm{d} y'  \right] b_y \, \mathrm{d} x \mathrm{d} y \\
& = \iiiint  {u(x)} \overline{u(y)} w_N(x-x') w_N(y-y') b_x^* [b_{x'}^*b_{x'},b_{y'}^* b_{y'}] b_y  \, \mathrm{d} x \mathrm{d} y \mathrm{d} x' \mathrm{d} y' \\
& = - \iiiint  {u(x)} \overline{u(y)} w_N(x-x') w_N(y-y') u(x')  \overline{u(y')}  b_x^* b_{x'}^*  b_{y'}  b_y \mathrm{d} x \mathrm{d} y \mathrm{d} x' \mathrm{d} y' \\
&\quad + \iiiint  {u(x)} \overline{u(y)} w_N(x-x') w_N(y-y')   u(y')  \overline{u(x')}  b_x^* b^*_{y'} b_{x'}  b_y \mathrm{d} x \mathrm{d} y \mathrm{d} x' \mathrm{d} y'
\end{align*}
The term with the minus sign is negative because
\begin{align*}
& \iiiint  {u(x)} \overline{u(y)} w_N(x-x') w_N(y-y') u(x')  \overline{u(y')}  b_x^* b_{x'}^*  b_{y'}  b_y \mathrm{d} x \mathrm{d} y \mathrm{d} x' \mathrm{d} y' \\
=  &
 \left( \iint  {u(x)} w_N(x-x') u(x') b_x^* b_{x'}^* \mathrm{d} x \mathrm{d} x' \right) \times \\
& \times \left( \iint  \overline{u(y)}w_N(y-y') \overline{u(y')} b_{y'}  b_y  \mathrm{d} y \mathrm{d} y'  \right)\\
 = &  AA^* \ge  0
\end{align*}
with 
$$
A= \iint  {u(x)} w_N(x-x') u(x') b_x^* b_{x'}^* \mathrm{d} x \mathrm{d} x' .
$$
Thus
\begin{align*}
& \iint  {u(x)} \overline{u(y)} b_x^* [B_x,B_y]b_y \mathrm{d} x \mathrm{d} y \\
& \le \iiiint  {u(x)} \overline{u(y)} w_N(x-x') w_N(y-y')   u(y')  \overline{u(x')}  b_x^* b^*_{y'} b_{x'}  b_y \mathrm{d} x \mathrm{d} y \mathrm{d} x' \mathrm{d} y'.
\end{align*}
Next, we use 
$$b_x^* b^*_{y'} b_{x'}  b_y=b_x^*b_{x'}b_{y'}^*b_y - Q(x',y')b_x^* b_y.$$
For the term involving $b_x^*b_{x'}b_{y'}^*b_y $
we have 
\begin{align*} 
\iiiint  {u(x)} \overline{u(y)} w_N(x-x') w_N(y-y')   u(y')  \overline{u(x')}  b_x^*b_{x'}b_{y'}^*b_y  \mathrm{d} x \mathrm{d} y \mathrm{d} x' \mathrm{d} y' = B^2
\end{align*}
where 
$$B:=\iint u(x)w_N(x-x')\overline{u(x')}b_x^* b_{x'} \mathrm{d} x \mathrm{d} x'.$$
By the Cauchy-Schwarz inequality and \eqref{eq:int-bx}, we can estimate
\begin{align*} 
\pm B &\le \frac{1}{2}\iint |w_N(x-x')|   \Big[ |u(x)|^2 b_x^*b_x + |u(x')|^2 b^*_{x'}b_{x'}  \Big] \mathrm{d} x \mathrm{d} x'\\
&\le C\|u(t)\|_{L^\infty}^2 \cN_+.
\end{align*}
Moreover, since $B$ commutes with $\cN_+$ we thus obtain
$$B^2\leq C\|u(t)\|_{L^\infty}^4 \cN_+^2(t).$$
It remains to bound the term involving $Q(x',y')b_x^* b_y$:
\begin{align*}
&\iiiint  {u(x)} \overline{u(y)} w_N(x-x') w_N(y-y')   u(y')  \overline{u(x')} Q(x',y')b_x^* b_y \mathrm{d} x \mathrm{d} y \mathrm{d} x' \mathrm{d} y' \\
=&\iiiint  {u(x)} \overline{u(y)} w_N(x-y') w_N(y-y')  | u(y')|^2 b_x^* b_y \mathrm{d} x \mathrm{d} y \mathrm{d} y' \\
&-\iiiint  {u(x)} \overline{u(y)} w_N(x-x') w_N(y-y')   |u(y')|^2  |u(x')|^2 b_x^* b_y   \mathrm{d} x \mathrm{d} y \mathrm{d} x' \mathrm{d} y'
\end{align*}
Now, by the Cauchy-Schwarz inequality and \eqref{eq:int-bx} again we have
\begin{align*}
& \pm\iiiint  {u(x)} \overline{u(y)} w_N(x-y') w_N(y-y')  | u(y')|^2 b_x^* b_y \mathrm{d} x \mathrm{d} y \mathrm{d} y'\\
&= \int \Big( \pm \iint  {u(x)} \overline{u(y)} w_N(x-y') w_N(y-y') b_x^* b_y \mathrm{d} x \mathrm{d} y \Big) |u(y')|^2 \mathrm{d} y'
\\
&\leq \int \frac{1}{2}\Big( \iint \Big[ |u(x)|^2 |w_N(y-y')|^2 b_x^* b_x + |u(y)|^2  |w_N(x-y')|^2\Big]    b_y^* b_y \mathrm{d} x \mathrm{d} y \Big) | u(y')|^2 \mathrm{d} y'\\
&\leq \|u(t)\|^2_{L^\infty} \|w_N\|_{L^2}^2 \cN_+(t)\\
&\leq CN^{3\beta} \|u(t)\|^2_{L^\infty} \cN_+ (t)
\end{align*}
and 
\begin{align*}
&\pm \iiiint  {u(x)} \overline{u(y)} w_N(x-x') w_N(y-y')   |u(y')|^2  |u(x')|^2 b_x^* b_y    \mathrm{d} x  \mathrm{d} y  \mathrm{d} x'  \mathrm{d} y' \\
&= \iint \Big(\pm \iint  u(x) \overline{u(y)} w_N(x-x') w_N(y-y')      b_x^* b_y    \mathrm{d} x  \mathrm{d} y  \Big) |u(x')|^2 |u(y')|^2 \mathrm{d} x'   \mathrm{d} y'\\
&\le \iint \frac{1}{2} \Big( \iint |w_N(x-x')| |w_N(y-y')| \Big[ |u(x)|^2 b_x^* b_x + |u(y)|^2 b_y^*b_y \Big]  \mathrm{d} x  \mathrm{d} y  \Big) \times \\
&\quad \times   |u(x')|^2 |u(y')|^2  \mathrm{d} x'   \mathrm{d} y'\\
& \leq \|w_N\|_{L^1}^2 \|u(t)\|^6_{L^\infty} \cN_+^2 (t).
\end{align*}

In summary, we have proved that 
\begin{align} \label{eq:R6*R6-ineq2}
\iint  {u(x)} \overline{u(y)} b_x^* [B_x,B_y]b_y \mathrm{d} x \mathrm{d} y  \le  CN^{3\beta} \cN^2_+(t).
\end{align}
From \eqref{eq:R6*R6}-\eqref{eq:R6*R6-ineq1}-\eqref{eq:R6*R6-ineq2} we find that
\begin{align} \label{eq:R6*R6-final}
R_6^* R_6 \le CN^{3\beta} \cN^3_+(t).
\end{align}
From \eqref{eq:R6R6*-final} and \eqref{eq:R6*R6-final}, it follows that
\begin{align}\label{eq:j=5}
&\quad\ R_{5} \1_+^{\le N} R_5^* +  R_5^* \1_+^{\le N} R_5 \\
& = \frac{\sqrt{N-\cN_+(t)}}{N-1} R_6 \1_+^{\le N}  R_6^* \frac{\sqrt{N-\cN_+(t)}}{N-1}  + R_6^*  \frac{(N-\cN_+(t))\1_+^{\le N}}{(N-1)^2} R_6 \nn\\
& \le \frac{\sqrt{N-\cN_+(t)}}{N-1} R_6 R_6^* \frac{\sqrt{N-\cN_+(t)}}{N-1} + \frac{N}{(N-1)^2} R_6^*R_6 \nn\\
& \le  C N^{3\beta-1} \cN_+^3(t).\nn
\end{align}

\noindent{\bf Conclusion.} Substituting \eqref{eq:j=1}, \eqref{eq:j=2}, \eqref{eq:j=3}, \eqref{eq:j=4} and \eqref{eq:j=5} into \eqref{eq:RN-R1-R5} we get the desired bound
\[
R^2(t)\le C \Big( N^{6\beta-2} \cN_+^4(t) + N^{3\beta-1}\cN_+^3(t) + N^{3\beta-2} \Big).\vspace{-1em}
\]
\end{proof}

\section{Evolution generated by quadratic Hamiltonians} \label{sec:evo-quasi-free}

\subsection{A general result} We have the following general result on the evolution generated by a quadratic Hamiltonian.

\begin{proposition}[Evolution generated by quadratic Hamiltonians] \label{lem:evol-quad-abstract} 
Let $\{\bH(t)\}$, $t\in [0,1]$, be a family of quadratic Hamiltonians on $\cF(\gH)$ of the form
$$ \bH(t) := \dGamma(h(t)) + \frac{1}{2} \iint \left( k(t,x,y)a^*_x a^*_y + \frac{1}{2} \overline{k(t,x,y)} a_x a_y \right) \mathrm{d} x  \mathrm{d} y$$
where $h(t)=h_1+h_2(t):\gH\to \gH$, with $h_1>0$ time-independent and $h_2(t)$ bounded, and with $k(t):\overline{\gH}\to \gH$ Hilbert-Schmidt with symmetric kernel\linebreak $k(t,x,y)=k(t,y,x)$. Assume that 
$$ \sup_{t\in [0,1]} \Big( \|h_2(t)\| + \| \partial_t h_2(t) \| + \|k(t,\cdot, \cdot)\|_{\gH^2} + \| \partial_t k(t,\cdot,\cdot) \|_{\gH^2} \Big) <\infty.$$
Then for every normalized vector $\Phi(0)\in \cF(\gH)$ satisfying $\langle \Phi(0), \cN \Phi(0)\rangle <\infty$, the equation
\be \label{eq:eq-Phit-abstract}
\begin{cases}
i \partial_t \Phi(t) =  \bH(t) \Phi(t), \\
\Phi(t=0) = \Phi(0)
\end{cases}
\ee
has a unique solution $\Phi(t)$ with $t\in [0,1]$ and the following statements hold true.
\begin{itemize}
\item[(i)] The pair of density matrices $(\gamma(t),\alpha(t))=(\gamma_{\Phi(t)},\alpha_{\Phi(t)})$ satisfies 
\be \label{eq:linear-eq-abstract}
\begin{cases}
i\partial_t \gamma = h \gamma - \gamma h + k\alpha^* - \alpha k^*, \\
i\partial_t \alpha =  h \alpha + \alpha h^{\rm T} + k  + k \gamma^{\rm T} + \gamma k, \\
\gamma(t=0) = \gamma(0), \quad \alpha(t=0)=\alpha(0).
\end{cases}
\ee
Moreover, $(\gamma_{\Phi(t)},\alpha_{\Phi(t)})$ is the unique solution to \eqref{eq:linear-eq-abstract} under the constraints 
\be 
\label{eq:assump-HS-alpha-gamma} \gamma=\gamma^*,~~\alpha=\alpha^{\rm T}, ~~ \sup_{t\in [0,1]} \Big( \Tr(\alpha(t)\alpha^*(t))  + \Tr(\gamma^2(t)) \Big) < \infty.
\ee

\item[(ii)] For every decomposition $k=k_1+k_2$, we have
\begin{align}\label{eq:bound-alpha-HS}
\|\alpha(t)\|_{\rm HS}^2+ \|\gamma(t)\|_{\rm HS}^2 \le \Theta(t) + \int_0^t \exp\Big( \int_s^t \xi(r) \mathrm{d} r \Big) \xi(s) \Theta (s) \mathrm{d} s
\end{align}
for all $t\in [0,1]$, where 
\begin{align*} 
\xi(t)&:=6\Big(\|k(t)\|+ \|h_2(t)\| \Big) ,\\
\Theta(t) &:= 2\|\alpha(0)\|_{\rm HS}^2 + 2\|\gamma(0)\|_{\rm HS}^2 \\
&\quad+ 2\bigg( \|\cL^{-1} k_1(t,\cdot, \cdot)\|_{L^2} + \|\cL^{-1} k_1(0,\cdot, \cdot)\|_{L^2}  \\[-1ex]
&\qquad\quad+\int_0^t \big( \|k_2(s,\cdot,\cdot)\|_{L^2}+ \|\cL^{-1} \partial_s k_1(s,\cdot, \cdot)\|_{L^2} \big) \mathrm{d} s \bigg)^2,\\
\cL &:= h_1\otimes 1 + 1\otimes h_1.
\end{align*}
\item[(iii)] If $\Phi(0)$ is a quasi-free state, then $\Phi(t)$ is a quasi-free state for all $t\in [0,1]$. In this case, 
\be \label{eq:bound-cN-abstract}
\langle \Phi(t), \cN \Phi(t) \rangle \le  \Theta_1(t) + \int_0^t \exp\left( \int_s^t \xi(r) \mathrm{d} r \right) \xi(s) \Theta_1 (s) \mathrm{d} s
\ee
where 
$$\Theta_1(t):=\Theta(t) -2\|\alpha(0)\|_{\rm HS}^2 - 2\|\gamma(0)\|_{\rm HS}^2 + 4\Big(1+\langle \Phi(0), \cN \Phi(0)\rangle \Big)^2.$$ 
\end{itemize}
\end{proposition}




Before proving Proposition \ref{lem:evol-quad-abstract}, let us recall some well-known properties of density matrices.  

\begin{lemma}[Density matrices] \label{lem:dm} Let $\Psi$ be a normalized vector in $\cF(\gH)$ with $\langle \Psi, \cN \Psi \rangle <\infty$ and let $(\gamma_\Psi,\alpha_\Psi)$ be its density matrices. Then 
\be \label{eq:gamma-alpha-ineq}
\gamma_\Psi \ge 0 , \quad  \Tr(\gamma_\Psi)=\langle \Phi, \cN \Phi\rangle, \quad \alpha = \alpha_\Psi^{\rm T}, \quad \gamma_\Psi \ge \alpha_\Psi (1+\gamma_\Psi^{\rm T})^{-1} \alpha_\Psi^*.
\ee
Moreover, $\Psi$ is a quasi-free if and only if
\be 
\label{eq:dm-quasi-free}
\gamma \alpha = \alpha \gamma^{\rm T}, \quad \alpha_\Psi \alpha_\Psi^* = \gamma_\Psi (1+\gamma_\Psi).
\ee
\end{lemma}

\begin{proof} The first three properties of \eqref{eq:gamma-alpha-ineq} are obvious. The last inequality of~\eqref{eq:gamma-alpha-ineq} is often formulated as  
$$
\Gamma:= \left( {\begin{array}{*{20}{c}}
   \gamma_\Psi & \alpha_\Psi \\ 
  \alpha_\Psi^{*} & 1+ \gamma_\Psi^{\rm T} 
\end{array}} \right) \ge 0
$$
on $\gH \oplus \overline{\gH}$, which can be seen immediately from the inequality
$$ \langle \Psi, (a(f)+a^*(g))^* (a(f)+a^*(g)) \Psi \rangle \ge 0 $$
for all $f,g\in \gH$ and the CCR \eqref{eq:ccr}. The equivalence between $\Gamma\ge 0$ and the last inequality in \eqref{eq:gamma-alpha-ineq} is straightforward, see e.g.  \cite[Lemma 1.1, p. 93]{Nam-thesis} for a proof. It is also well-known, see e.g. \cite[Theorem 10.4]{Solovej-ESI-2014}, that $\Psi$ is a quasi-free if and only if 
$$
\Gamma \left( {\begin{array}{*{20}{c}}
   1 & 0 \\ 
  0 & -1 
\end{array}} \right) \Gamma = - \Gamma.
$$
A direct calculation shows that the latter equality is equivalent to \eqref{eq:dm-quasi-free}. 
\end{proof}
\begin{remark} In fact, the first condition $\gamma \alpha = \alpha \gamma^{\rm T}$ of \eqref{eq:dm-quasi-free} is a consequence of the second condition $\alpha_\Psi \alpha_\Psi^* = \gamma_\Psi (1+\gamma_\Psi)$, see \cite[Proposition 4.5]{BacBreKonMen-14}.
\end{remark}

Now we are able to give

\begin{proof}[Proof of Proposition  \ref{lem:evol-quad-abstract}]
 
The proof is divided into several steps.
\smallskip

\noindent {\bf Step 1} (Existence and uniqueness of solution to \eqref{eq:eq-Phit-abstract}). Note that from the  estimate \eqref{eq:Kcr*-Kcr} and the operator monotonicity of the square root, we get
$$
\pm \left( \iint \Big(k(t,x,y)a^*_xa^*_y + {\rm h.c.}\Big) \mathrm{d} x \mathrm{d} y \right) \le \sqrt{2} \| k(t,\cdot,\cdot) \|_{L^2} (\cN+1).
$$
Recall that we have denoted $X+{\rm h.c.}=X+X^*$ for short. Therefore, from the conditions on $h(t)$ and $k(t)$ we get
\begin{align*}
\pm (\bH(t) -\dGamma(h_1)) &= \pm \left( \dGamma(h_2(t)) +  \iint \Big(k(t,x,y)a^*_xa^*_y + {\rm h.c.}\Big) \mathrm{d} x \mathrm{d} y  \right) \\
&\le C_2 (\cN +1),\\
\pm \partial_t \bH(t) &= \pm \left( \dGamma( \partial_t h_2(t)) +  \iint \Big(\partial_t k(t,x,y)a^*_xa^*_y + {\rm h.c.}\Big) \mathrm{d} x \mathrm{d} y \right)\\
&\le C_2 (\cN +1),\\
\pm i [\bH(t),\N ] &= \pm \frac{i}{2}\iint \Big[k(t,x,y)a^*_x a^*_y + {\rm h.c.},\cN \Big] \mathrm{d} x \, \mathrm{d} y \nn\\
&=  \pm \iint \Big( - i k(t,x,y) a^*_x a^*_y  + {\rm h.c.} \Big) \mathrm{d} x \, \mathrm{d} y\\
&\le C_2(\cN + 1)
\end{align*}
with 
\be \label{eq:def-C2}
C_2:=\sup_{t\in [0,1]} \left( \|h_2(t)\| + \sqrt{2}\|k(t, \cdot, \cdot)\|_{L^2} \right) <\infty. 
\ee
Thus by \cite[Theorem 8]{LewNamSch-14}, we know that for every $\Phi(0)$ in the quadratic form domain of $\dGamma(h_1+1)$, the equation \eqref{eq:eq-Phit-abstract} has a unique solution $\Phi(t)$. 

Note that for every time-independent operator $\mathcal{O}$ on Fock space, 
\begin{align} \label{eq:diff-observable}
i \partial_t \langle \Phi(t), \mathcal{O} \Phi(t) \rangle &= - \langle i \partial_t \Phi(t), \mathcal{O} \Phi(t) \rangle + \langle  \Phi(t), \mathcal{O} i \partial_t \Phi(t)  \\
&= \langle \Phi(t), [\mathcal{O}, \bH(t)] \Phi(t) \rangle.\nn
\end{align}
In particular, choosing $\mathcal{O}=1$ we get $\|\Phi(t)\|=\|\Phi(0)\|$. Therefore, the  propagator $\mathscr{U}(t): \cF(\gH)\to \cF(\gH)$ defined by $\mathscr{U}(t) \Phi(0)=\Phi(t)$ is a (partial) unitary, and hence it can be extended from the quadratic form domain of $\dGamma(h_1+1)$ to the whole Fock space. In the following, we are interested in the case when $\langle \Phi(0), \cN \Phi(0) \rangle<\infty$. In this case, by \eqref{eq:diff-observable} we have 
$$
\frac{d}{dt} \langle \Phi(t), (\cN+1) \Phi(t) \rangle = \langle \Phi(t), i[\bH(t),\cN ] \Phi(t)\rangle \le C_2 \langle \Phi(t), (\cN+1) \Phi(t) \rangle, 
$$
and by Gr\"onwall's inequality, 
\be \label{eq:Phit-cN-weak-bound}
\langle \Phi(t), (\cN+1) \Phi(t) \rangle \le e^{C_2 t} \langle \Phi(0), (\cN+1) \Phi(0) \rangle.
\ee

Thus, in summary, for every $\Phi(0)$ with $\langle \Phi(0), \cN \Phi(0) \rangle <\infty$, the equation~\eqref{eq:eq-Phit-abstract} has a unique solution $\Phi(t)$ satisfying $\|\Phi(t)\|=\|\Phi(0)\|$ and \eqref{eq:Phit-cN-weak-bound}. 
\medskip

\noindent {\bf Step 2} (Derivation of linear equations \eqref{eq:linear-eq-abstract}). We will use \eqref{eq:diff-observable} to compute the time-derivatives of the kernels
$$ 
\gamma_{\Phi(t)}(x,y)=\langle \Phi(t), a_y^* a_x \Phi(t) \rangle, \quad \alpha_{\Phi(t)}(x,y)=\langle \Phi(t), a_x a_y \Phi(t) \rangle.
$$

Let us use \eqref{eq:diff-observable} with $\mathcal{O}=a_{y'}^* a_{x'}$. From the CCR \eqref{eq:ccr-x} it follows that
\begin{align*}
[a_{y'}^* a_{x'}, \bH(t)] &= \iint \bigg[a_{y'}^* a_{x'}, h(t,x,y)a^*_x a_y\\
& \qquad \quad\ + \frac{1}{2} k(t,x,y)a^*_x a^*_y + \frac{1}{2} k^*(t,x,y) a_x a_y \bigg] \mathrm{d} x  \mathrm{d} y \\
& = \iint h(t,x,y) \Big(\delta (x'-x) a^*_{y'}a_y - \delta (y'-y)a^*_x a_{x'}) \Big) \mathrm{d} x  \mathrm{d} y \\
&\quad + \frac{1}{2} \iint k(t,x,y) \Big(  \delta(x'-x) a^*_{y'}a^*_y + \delta (x'-y)a^*_{y'}a^*_x \Big) \mathrm{d} x  \mathrm{d} y \\
&\quad - \frac{1}{2} \iint k^*(t,x,y) \Big(\delta(y'-y) a_x a_{x'} + \delta (y'-x) a_y a_{x'} \Big) \mathrm{d} x  \mathrm{d} y.
\end{align*}
Therefore, 
\begin{align*}
&\quad\ i\partial_t \gamma_{\Phi(t)}(x',y') = i\partial_t \langle \Phi(t), a_{y'}^* a_x \Phi(t) \rangle = \langle \Phi(t), [a_{y'}^* a_x , \bH(t)] \Phi(t) \rangle \\
&= \iint h(t,x,y) \Big(\delta (x'-x) \gamma_{\Phi(t)}(, y,y') - \delta (y'-y) \gamma_{\Phi(t)}(x',x) \Big) \mathrm{d} x  \mathrm{d} y \\
&\quad + \frac{1}{2} \iint k(t,x,y) \Big( \delta(x'-x) \alpha_{\Phi(t)}^*(y,y') + \delta (x'-y) \alpha_{\Phi(t)}^*(y',x) \Big) \mathrm{d} x  \mathrm{d} y \\
&\quad - \frac{1}{2} \iint k^*(t,x,y) \Big( \delta(y'-y) \alpha_{\Phi(t)}(,x,x') + \delta (y'-x) \alpha_{\Phi(t)}(,y,x') \Big) \mathrm{d} x  \mathrm{d} y \\
& = \Big( h(t) \gamma_{\Phi(t)} - \gamma_{\Phi(t)} h(t) + k(t)\alpha^*_{\Phi(t)} - \alpha_{\Phi(t)} k^*(t) \Big)(x',y'). 
\end{align*}
Here we have used $k(t,x,y)=k(t,y,x)$ and $\alpha_{\Phi(t)}(x,y)=\alpha_{\Phi(t)}(y,x)$. Similarly, using \eqref{eq:diff-observable} with $\mathcal{O}= a_{x'} a_{y'}$ and the identity

\begin{align*}
[a_{x'} a_{y'}, \bH(t)] &= \iint \Big[a_{x'} a_{y'}, h(t,x,y)a^*_x a_y \\
& \quad\qquad\ + \frac{1}{2} k(t,x,y)a^*_x a^*_y + \frac{1}{2} k^*(t,x,y) a_x a_y \Big] \mathrm{d} x  \mathrm{d} y \\
& =  \iint h(t,x,y) \Big( \delta (y'-x) a_{x'} a_{y} + \delta (x'-x)a_{y} a_{y'} \Big) \mathrm{d} x  \mathrm{d} y \\
& \quad + \frac{1}{2} \iint k(t,x,y) \Big( \delta(y'-x)\delta(x'-y) + \delta(y'-x) a^*_y a_{x'} \\
& \qquad \qquad\qquad\qquad\quad + \delta(x'-x)\delta(y'-y) + \delta(x'-x)  a^*_y a_{y'} \\
& \qquad \qquad\qquad\qquad\quad  + \delta(y'-y) a^*_{x} a_{x'} + \delta(x'-y) a^*_{x} a_{y'} \Big) \mathrm{d} x  \mathrm{d} y,
\end{align*}
we find that
\begin{align*}
&\quad\ i\partial_t \alpha_{\Phi(t)}(x',y') = i\partial_t \langle \Phi(t), a_{x'}^* a^*_{y'} \Phi(t) \rangle = \langle \Phi(t), [a_{x'}^* a^*_{y'} , \bH(t)] \Phi(t) \rangle \\
&=  \iint h(t,x,y) \Big( \delta (y'-x) \alpha_{\Phi(t)}(x',y)  + \delta (x'-x)\alpha_{\Phi(t)}(t,y,y')  \Big) \mathrm{d} x  \mathrm{d} y \\
& \quad + \frac{1}{2} \iint k(t,x,y) \Big(  \delta(y'-x)\delta(x'-y) + \delta(y'-x) \gamma_{\Phi(t)}(x',y)   \\
& \qquad \qquad\qquad \qquad\quad + \delta(x'-x)\delta(y'-y) + \delta(x'-x)  \gamma_{\Phi(t)}(y',y)    \\
& \qquad \qquad\qquad \qquad\quad + \delta(y'-y) \gamma_{\Phi(t)}(x',x) + \delta(x'-y) \gamma_{\Phi(t)}(y',x)  \Big) \mathrm{d} x  \mathrm{d} y, \\
& =  \Big(\alpha_{\Phi(t)} h^{\rm T}(t) + h(t) \alpha_{\Phi(t)} + k(t) + \gamma_{\Phi(t)} k(t) + k(t) \gamma^{\rm T}_{\Phi(t)} \Big)(x',y').
\end{align*}
Thus $(\gamma_{\Phi(t)},\alpha_{\Phi(t)})$ satisfies the couple of linear equations \eqref{eq:linear-eq-abstract}:
$$
\begin{cases}
i\partial_t \gamma = h \gamma - \gamma h + k\alpha^* - \alpha k^*, \\
i\partial_t \alpha =  h \alpha + \alpha h^{\rm T} + k  + k \gamma^{\rm T} + \gamma k,\\
\end{cases}
$$
with initial data $\gamma(0) = \gamma_{\Phi(0)}$, $\alpha(0)=\alpha_{\Phi(0)}$.
\medskip

\noindent {\bf Step 3} (Uniqueness of solution to \eqref{eq:linear-eq-abstract}). By the derivation in the previous step, we have proved that $(\gamma_{\Phi(t)},\alpha_{\Phi(t)})$ is a solution to \eqref{eq:linear-eq-abstract}. Moreover, from~\eqref{eq:Phit-cN-weak-bound} and \eqref{eq:gamma-alpha-ineq} we can see that if $\langle \Phi(0), \cN \Phi(0) \rangle<\infty$, then $\Tr(\gamma_{\Phi(t)})$ and $\Tr(\alpha_{\Phi(t)} \alpha^*_{\Phi(t)})$ are bounded uniformly in $t\in [0,1]$. 

Now we show that for every given initial condition $(\gamma(0),\alpha(0))$, the system \eqref{eq:linear-eq-abstract} has at most one solution satisfying \eqref{eq:assump-HS-alpha-gamma}:
$$ 
\gamma=\gamma^*,\quad \alpha=\alpha^{\rm T}, \quad \sup_{t\in [0,1]} \Big( \Tr(\alpha(t)\alpha^*(t))  + \Tr(\gamma^2(t)) \Big) < \infty.
$$
More precisely, we will show that if $(\gamma_j(t), \alpha_j(t))_{j=1,2}$ are two solutions to~\eqref{eq:linear-eq-abstract} with the same initial condition and they satisfy \eqref{eq:assump-HS-alpha-gamma}, then $$X(t):=\gamma_1(t)-\gamma_2(t) \quad {\rm and} \quad Y(t):=\alpha_1(t)-\alpha_2(t)$$
are $0$ for all $t\in [0,1]$. It is clear that $X,Y$ satisfy 
\be \label{eq:linear-eq-homo}
\begin{cases}
i\partial_t X = h X - X h + kY^* - Y k^*, \\
i\partial_t Y =  h Y + Y h^{\rm T} + k X^{\rm T} + X k, \\
X(0) = 0, \quad Y(0)=0
\end{cases}
\ee
and
\be \label{eq:assump-HS-X-Y} X=X^*, \quad Y=Y^{\rm T}, \quad \sup_{t\in [0,1]} \Big( \Tr(Y(t)Y^*(t))  + \Tr(X^2(t)) \Big) < \infty.
\ee
Note that the second equation of \eqref{eq:linear-eq-homo} is different from that of \eqref{eq:linear-eq-abstract} because the inhomogeneous term $k$ has been canceled. 

From the first equation of \eqref{eq:linear-eq-homo} we have
\begin{align} \label{eq:dt-XX}
i\partial_t X^2 &= (i\partial_t X) X + X (i\partial_t X) \\
& = (h X - X h + kY^* - Y k^*)X + X (h X - X h + kY^* - Y k^*) \nn \\
&= hX^2 - X^2 h + (kY^* - Y k^*)X + X(kY^* - Y k^*).\nn
\end{align}
We want to take the trace of \eqref{eq:dt-XX} and use the cancellation $\Tr(h X^2 - X^2 h)=0$  but it is a bit formal because $h X^2$ and $X^2 h$ might be not trace class. To make the argument rigorous, let us introduce the time-independent projection $\1^{\le \Lambda}=\1(h_1\le \Lambda)$ with $0<\Lambda<\infty$ and deduce from \eqref{eq:dt-XX} that
\begin{align} \label{eq:XX-proj}
i \partial_t (\1^{\le \Lambda} X^2 \1^{\le \Lambda} ) = \1^{\le \Lambda} (hX^2 - X^2 h  + {\rm Er}) \1^{\le \Lambda}
\end{align}
where
$$
{\rm Er}:=(kY^* - Y k^*)X + X(kY^* - Y k^*).
$$
Now we can take the trace of both sides of \eqref{eq:XX-proj} and then integrate over $t$. We will use the cyclicity of the trace
$$\Tr(\mathfrak{X}_1 \mathfrak{X}_2)=\Tr(\mathfrak{X}_2\mathfrak{X}_1)$$
with $\mathfrak{X}_1$ bounded and $\mathfrak{X}_2$ trace class. In particular, 
\begin{align*} &\quad \Tr(\1^{\le \Lambda} (h_1 X^2 -X^2 h_1) \1^{\le \Lambda})\\
&= \Tr(\1^{\le \Lambda} h_1 \cdot \1^{\le \Lambda}  X^2  \1^{\le \Lambda}) - \Tr(\1^{\le \Lambda}  X^2  \1^{\le \Lambda} \cdot \1^{\le \Lambda} h_1 )=0.
\end{align*}
Therefore, \eqref{eq:XX-proj} gives us
\begin{align} \label{eq:XX-proj-int} \Tr  (\1^{\le \Lambda} X^2(t) \1^{\le \Lambda} ) = - i \int_0^t \Tr \Big[ \1^{\le \Lambda} \Big( h_2 X^2- X^2 h_2 + {\rm Er} \Big) \1^{\le \Lambda} \Big](s) \mathrm{d} s.
\end{align}
Next, we pass $\Lambda\to +\infty$ and use the convergence
$$
\lim_{\Lambda\to \infty}\Tr(\1^{\le \Lambda} \mathfrak{X} \1^{\le \Lambda}) =\Tr(\mathfrak{X})$$
which holds for every trace class operator $\mathfrak{X}$. Note that 
\begin{align*} 
\Tr \Big| h_2 X^2 - X^2 h_2 \Big| \le  \|h_2\|\cdot \|X\|_{\rm HS}^2, \quad \Tr \Big| {\rm Er} \Big| \le 4 \|k\| \cdot \|X\|_{\rm HS} \cdot \|Y\|_{\rm HS}
\end{align*}
and $\|h_2(t)\|$, $\|k(t)\|$, $\|X(t)\|_{\rm HS}$, $\|Y(t)\|_{\rm HS}$ are bounded uniformly in $t\in [0,1]$. Therefore, we can use Lebesgue's Dominated Convergence Theorem and deduce from \eqref{eq:XX-proj-int} that
$$
\|X(t)\|_{\rm HS}^2 = - i \int_0^t \Big[ \Tr \Big( X^2 h_2 -  h_2 X^2 + {\rm Er} \Big) \Big] (s) \mathrm{d} s.
$$
Using again the cyclicity of the trace, $\Tr( h_2 X^2)=\Tr(X^2 h_2)$, we get  
\begin{align} \label{eq:XX-int}
\|X(t)\|_{\rm HS}^2 = - i \int_0^t {\rm Er}(s) \mathrm{d} s \le 4 \int_0^t \|k(s)\|\cdot \|X(s)\|_{\rm HS}\cdot \|Y(s)\|_{\rm HS} \,  \mathrm{d} s. 
\end{align}
Similarly, from the second equation of \eqref{eq:linear-eq-homo} we have
\begin{align*} 
i \partial_t (Y Y^*) &= (i\partial_t Y) Y^* - Y (i \partial_t Y)^* \nn\\
& = (h Y + Y h^{\rm T} +  k X^{\rm T} + X k)Y^* - Y (h Y + Y h^{\rm T} +  k X^{\rm T} + X k)^* \nn\\
& = h Y Y^* - Y Y^* h + (k X^{\rm T} + X k)Y^* -Y (k X^{\rm T} + X k)^*,
\end{align*}
and hence 
\begin{align} \label{eq:YY*-int}
\|Y(t)\|_{\rm HS}^2 &= - i \int_0^t \Big[ (k X^{\rm T} + X k)Y^* -Y (k X^{\rm T} + X k)^* \Big](s) \mathrm{d} s \\
&\le 4 \int_0^t \|k(s)\|\cdot \|X(s)\|_{\rm HS}\cdot \|Y(s)\|_{\rm HS} \,  \mathrm{d} s.\nn 
\end{align}
Summing \eqref{eq:XX-int} and \eqref{eq:YY*-int}  we find that
\begin{align*}
\|X(t)\|_{\rm HS}^2 + \|Y(t)\|_{\rm HS}^2 &\le 8\int_0^1 \|k(s)\|   \cdot \|X(s)\|_{\rm HS} \cdot \|Y(s)\|_{\rm HS} \, \mathrm{d} s \\
& \le 4 \left( \sup_{r\in [0,1]} \|k(r)\| \right)  \int_0^t (\|X(s)\|_{\rm HS}^2 + \|Y(s)\|_{\rm HS}^2)
\end{align*}
for all $t\in [0,1]$. Consequently, $X(t)\equiv 0$ and $Y(t)\equiv 0$ by Gr\"onwall's inequality. This finishes the proof of the uniqueness.

Thus $(\gamma(t),\alpha(t))=(\gamma_{\Phi(t)}, \alpha_{\Phi(t)})$ is the unique solution to the system \eqref{eq:linear-eq-abstract} under conditions \eqref{eq:assump-HS-alpha-gamma}. 
\medskip

\noindent {\bf Step 4} (Improved bound on $\|\alpha\|_{\rm HS}$) Recall that from \eqref{eq:Phit-cN-weak-bound} and \eqref{eq:gamma-alpha-ineq} we already have upper bounds on $\Tr(\gamma(t))$ and $\|\alpha\|_{\rm HS}$. However, the constant $C_2$ defined in \eqref{eq:def-C2} depends on $\|k(t)\|_{\rm HS}=\|k(t,\cdot, \cdot)\|_{L^2}$ which is large in our application. In the following we will derive another bound on $\|\alpha\|_{\rm HS}$ which depends on the operator norm $\|k(t)\|$ instead of the Hilbert-Schmidt norm.

Inspired by \cite[Proof of Theorem 4.1]{GriMac-13}, we will decompose $\alpha(t)=Y_1(t)+Y_2(t)$ where $Y_1(t),Y_2(t):\overline{\gH} \to \gH$ satisfy
\be \label{eq:Y_1eq}
\begin{cases}
i\partial_t Y_1= h_1 Y_1 + Y_1 h_1^{\rm{T}}  + k, \\
Y_1(t=0)=0,
\end{cases}
\ee
and
\be \label{eq:Y_2Xeq}
\begin{cases}
i\partial_t \gamma = h \gamma - \gamma h + k(Y_1+Y_2)^* - (Y_1+Y_2) k^*,\\
i\partial_t Y_2= hY_2+Y_2 h^{\rm T} + h_2 Y_1 + Y_1 h_2^{\rm T} + k \gamma^{\rm T} + \gamma k,\\
Y_2(t=0)=\alpha(0), \quad \gamma(t=0)=\gamma(0).&
\end{cases}
\ee
Here \eqref{eq:Y_2Xeq} is derived from \eqref{eq:Y_1eq} and the system \eqref{eq:linear-eq-abstract}. 
\medskip

\noindent{\bf Estimation of $\|Y_1\|_{\rm HS}$.} Note that the kernel of the operator  
$h_1 Y_1: \overline{\gH} \to \gH$ is the two-body function
$$
\int h_1(x,z) Y_1(t,z,y) \mathrm{d} z = \Big[ (h_1\otimes 1) Y_1(t, \cdot, \cdot) \Big](x,y).
$$
The latter equality follows from a straightforward calculation:
\begin{align*}
&\quad\ \iint \overline{f(x)g(y)} \Big[(h_1\otimes 1) Y_1(t, \cdot, \cdot) \Big] (x,y) \mathrm{d} x \mathrm{d} y \\
&= \Big\langle f\otimes g, (h_1\otimes 1) Y_1 (t, \cdot, \cdot) \Big\rangle_{L^2} = \Big\langle (h_1\otimes 1) (f\otimes g),  Y_1 (t, \cdot, \cdot) \Big\rangle_{L^2} \\
&= \iiint \overline{h_1(x,z)f(z)g(y)}Y_1(t,x,y) \mathrm{d} x \mathrm{d} y \mathrm{d} z \\
 & = \iiint h_1(z,x)\overline{f(z)g(y)} Y_1(t,x,y) \mathrm{d} x \mathrm{d} y \mathrm{d} z \\
&=\iiint  h_1(x,z) \overline{f(x)g(y)} Y_1(t,z,y) \mathrm{d} x \mathrm{d} y \mathrm{d} z 
\end{align*}
for all $f,g\in \gH$. Here we have used the fact that $h_1$ is self-adjoint, which implies that $h_1(x,y)=\overline{h_1(y,x)}$. Similarly, the kernel of $Y_1 h_1^{\rm T}=(h_1 Y_1)^{\rm T}$ is $\Big[ (1\otimes h_1) Y_1(t, \cdot, \cdot) \Big](x,y)$. Therefore, the operator equation \eqref{eq:Y_1eq} can be rewritten as an equation of two-body functions:
\be \label{eq:Y_1eq-functions}
\begin{cases}
i\partial_t Y_1(t,x,y) = (\mathcal{L}  Y_1)(t,x,y) +  k(t,x,y), \\
Y_1(0,x,y)=0.
\end{cases}
\ee
where
$$ \cL :=  h_1\otimes 1 + 1\otimes h_1 >0 .$$
By Duhamel's formula and integration by parts, we can write
\begin{align*}
Y_1(t,x,y)&= \int_0^t e^{i(s-t) \cL } k(s,x,y) \mathrm{d} s \\
&= \int_0^t e^{i(s-t) \cL } k_2(s,x,y) \mathrm{d} s -i\int_0^t \partial_s (e^{i(s-t) \cL} \cL ^{-1}) k_1(s,x,y) \mathrm{d} s\\
& = - i \cL^{-1} k_1(t,x,y) + i e^{-it\cL} \cL^{-1} k_1(0,x,y) \\
&\quad + \int_0^t e^{i(s-t) \cL} \Big( k_2(s,x,y) + i\cL ^{-1} \partial_s k_1(s,x,y) \Big) \mathrm{d} s.
\end{align*}
Since $e^{-it\cL}$ is a unitary operator on $\gH^2$, by the triangle inequality we get
\begin{align} \label{eq:Y1-ineq-final}
 \|Y_1(t)\|_{\rm HS} &= \|Y_1(t, \cdot, \cdot ) \|_{L^2} \\
&\le  \|\cL^{-1} k_1(t,\cdot, \cdot)\|_{L^2} + \|\cL^{-1} k_1(0,\cdot, \cdot)\|_{L^2}\nn \\
&\quad +\int_0^t \Big( \|k_2(t,\cdot,\cdot)\|_{L^2}+\|\cL^{-1} \partial_s k_1(s,\cdot, \cdot)\|_{L^2} \Big) \mathrm{d} s. \nn
\end{align}

\noindent{\bf Estimation of $\|\gamma\|_{\rm HS}$ and $\|Y_2\|_{\rm HS}$ from \eqref{eq:Y_2Xeq}.} We use again the argument of deriving \eqref{eq:XX-int} in Step 3. From the first  equation of \eqref{eq:Y_2Xeq} we have
\begin{align*}
i \partial_t \gamma^2 &= (i\partial_t \gamma) \gamma + \gamma (i\partial_t \gamma) \\
&= \Big( h\gamma-\gamma h+k (Y_1+Y_2)^* -(Y_1+Y_2) k^* \Big)\gamma \\
&\quad + \gamma \Big( h\gamma-\gamma h+k (Y_1+Y_2) ^* -(Y_1+Y_2) k^* \Big) \\
&= h \gamma^2 - \gamma^2 h + k(Y_1+Y_2)^*\gamma - (Y_1+Y_2) k^*\gamma \\
&\quad + \gamma k (Y_1+Y_2)^* - \gamma ( Y_1+Y_2) k^*,
\end{align*}
and hence
\begin{align} \label{eq:gamma-gamma-int}
&\quad\ \|\gamma(t)\|_{\rm HS}^2 - \|\gamma(0)\|_{\rm HS}^2 \\
&= - i \int_0^t \Tr \Big( k(Y_1+Y_2)^*\gamma - (Y_1+Y_2) k^*\gamma \nn\\
& \qquad\qquad\qquad  + \gamma k (Y_1+Y_2)^* - \gamma ( Y_1+Y_2) k^* \Big) (s) \mathrm{d} s \nn\\
&\le 4 \int_0^t \Big( \|k\| \cdot \|\gamma\|_{\rm HS} \|Y_1\|_{\rm HS} +  \|k\| \cdot \|\gamma\|_{\rm HS}  \|Y_2\|_{\rm HS} \Big)(s) \mathrm{d} s. \nn
\end{align}
From the second equation of \eqref{eq:Y_2Xeq} we have
\begin{align*} 
i \partial_t (Y_2Y_2^*) &= (i\partial_t Y_2) Y_2^* - Y_2 (i\partial_t Y_2)^*  \nn \\
& = (hY_2+Y_2 h^{\rm T} + h_2 Y_1 + Y_1 h_2^{\rm T} + k \gamma^{\rm T} + \gamma k)Y_2^* \\
& \quad - Y_2 (hY_2+Y_2 h^{\rm T} + h_2 Y_1 + Y_1 h_2^{\rm T} + k \gamma^{\rm T} + \gamma k)^* \\
& = hY_2Y_2^* - Y_2Y_2^* h + \left(h_2 Y_1 + Y_1 h_2^{\rm T} + k \gamma^{\rm T} + \gamma k \right) Y_2^*  \\
& \quad - Y_2 \left( h_2 Y_1 + Y_1 h_2^{\rm T} + k \gamma^{\rm T} + \gamma k \right)^*, 
\end{align*}
and hence 
\begin{align} \label{eq:Y2Y2*-int}
&\quad\ \|Y_2(t)\|_{\rm HS}^2 - \|\alpha(0)\|_{\rm HS}^2\\
&= - i \int_0^t  \Tr \Big( \left(h_2 Y_1 + Y_1 h_2^{\rm T} + k \gamma^{\rm T} + \gamma k \right) Y_2^* \nn\\
&\qquad \qquad\qquad+ Y_2 \left( h_2 Y_1 + Y_1 h_2^{\rm T} + k \gamma^{\rm T} + \gamma k \right)^* \Big) (s) \mathrm{d} s \nn\\
&\le 4 \int_0^t \Big( \|h_2\| \cdot \| Y_1\|_{\rm HS}  \|Y_2\|_{\rm HS} + \|k\|\cdot \|\gamma\|_{\rm HS}  \|Y_2\|_{\rm HS} \Big) (s) \mathrm{d} s. \nn
\end{align}
{\bf Conclusion.} Summing \eqref{eq:Y1-ineq-final}, \eqref{eq:gamma-gamma-int} and \eqref{eq:Y2Y2*-int}, we get
\begin{align} \label{eq:Y*Y-ineq-final}
&\quad\ \|Y_1(t)\|_{\rm HS}^2 + \|Y_2(t)\|_{\rm HS}^2 + \|\gamma(t)\|_{\rm HS}^2 \\
&\le \frac{1}{2}\Theta(t) + 4 \int_0^t \Big( \|k\| \cdot \|\gamma\|_{\rm HS} \|Y_1\|_{\rm HS} \nn\\
&\hspace{6em} + 2\|k\| \cdot \|\gamma\|_{\rm HS}  \|Y_2\|_{\rm HS} + \|h_2\| \cdot \| Y_1\|_{\rm HS}  \|Y_2\|_{\rm HS}\Big)(s) \mathrm{d} s \nn\\
&\le \int_0^t 6\Big(\|k(s)\|+\|h_2(s)\| \Big) \Big( \|Y_1(s)\|_{\rm HS}^2 + \|Y_2(t)\|_{\rm HS}^2 +   \|\gamma(t)\|_{\rm HS}^2 ) \mathrm{d} s\nn
\end{align}
where
\begin{align*} 
\Theta(t) &:= 2\|\gamma(0)\|_{\rm HS}^2 + 2\|\alpha(0)\|_{\rm HS}^2 \\
&\quad+ 2\Bigg( \|\cL^{-1} k_1(t,\cdot, \cdot)\|_{L^2} + \|\cL^{-1} k_1(0,\cdot, \cdot)\|_{L^2}  \\[-1ex]
&\qquad\quad+\int_0^t \big( \|k_2(s,\cdot,\cdot)\|_{L^2}+ \|\cL^{-1} \partial_s k_1(s,\cdot, \cdot)\|_{L^2} \big) \mathrm{d} s \Bigg)^2.
\end{align*}
Note that if $f,g,\xi:\R_+\to \R_+$ satisfy 
$$ f(t) \le g(t) + \int_0^t \xi(s) f(s) \mathrm{d} s$$
for all $t$, then we have the Gr\"onwall-type inequality
\begin{align*} f(t) \le g(t) +  \int_0^t \exp\left( \int_s^t \xi(r) \mathrm{d} r \right) \xi(s) g(s) \mathrm{d} s.
\end{align*}
Therefore, we can deduce from \eqref{eq:Y*Y-ineq-final} that
\begin{align} \label{eq:alpha-HS-gamma-final}
\|\alpha(t)\|_{\rm HS}^2+ \|\gamma(t)\|_{\rm HS}^2 &\le 2\Big(\|Y_1(t)\|_{\rm HS}^2 + \|Y_2(t)\|_{\rm HS}^2) + \|\gamma(t)\|_{\rm HS}^2 \Big)\\
&\le \Theta(t) + \int_0^t \exp\left( \int_s^t \xi(r) \mathrm{d} r \right) \xi(s) \Theta (s) \mathrm{d} s.\nn
\end{align}
where 
$$ 
\xi(t):=6\Big(\|k(t)\|+\|h_2(t)\| \Big).
$$

\noindent {\bf Step 5} (Quasi-free states). In this final step we prove that if $\Phi(0)$ is a quasi-free state, then $\Phi(t)$ is a quasi-free for all $t\in [0,1]$. By Lemma \ref{lem:dm}, it suffices to show that the density matrices $(\gamma(t),\alpha(t))$ of $\Phi(t)$ satisfy 
\be \label{eq:dm-qf}
\gamma \alpha = \alpha \gamma^{\rm T}, \quad \alpha \alpha ^* =\gamma(1+\gamma).
\ee
Indeed, using the equations \eqref{eq:linear-eq-abstract}, we can compute
\begin{align*}
&\quad\ i \partial_t ( \gamma + \gamma^2 - \alpha \alpha^*)= (i\partial_t \gamma) (1+\gamma) + \gamma(i \partial_t \gamma) - (i\partial_t \alpha) \alpha^* + \alpha (i\partial_t \alpha)^* \\
& = ( h \gamma - \gamma h + k\alpha^* - \alpha k^* ) (1+\gamma) + \gamma (h \gamma - \gamma h + k\alpha^* - \alpha k^*) \\
& \quad - ( h \alpha + \alpha h^{\rm T} + k  + k \gamma^{\rm T} + \gamma k) \alpha^* + \alpha ( h \alpha + \alpha h^{\rm T} + k  + k \gamma^{\rm T} + \gamma k)^* \\
& = h(\gamma+\gamma^2-\alpha\alpha^*) -  (\gamma+\gamma^2-\alpha\alpha^*) h + k(\alpha^*\gamma -\gamma^{\rm T} \alpha^*)- (\gamma \alpha - \alpha \gamma^{\rm T}) k^*
\end{align*}
and
\begin{align*}
&\quad\ i \partial_t (\gamma \alpha - \alpha \gamma^{\rm T}) = (i\partial_t \gamma)\alpha + \gamma (i\partial_t \alpha) - (i\partial_t \alpha) \gamma^{\rm T} - \alpha (i\partial_t \gamma)^{\rm T} \\
& = (h \gamma - \gamma h + k\alpha^* - \alpha k^*)\alpha + \gamma ( h \alpha + \alpha h^{\rm T} + k  + k \gamma^{\rm T} + \gamma k ) \\
&\quad - (h \alpha + \alpha h^{\rm T} + k  + k \gamma^{\rm T} + \gamma k ) \gamma^{\rm T} - \alpha (h \gamma - \gamma h + k\alpha^* - \alpha k^*)^{\rm T} \\
& = h(\gamma \alpha - \alpha \gamma^{\rm T} ) + (\gamma \alpha - \alpha \gamma^{\rm T}) h - k (\gamma + \gamma^2 -\alpha \alpha^*)^{\rm T} + (\gamma + \gamma^2 -\alpha \alpha^*) k.
\end{align*}
The latter two equations can be rewritten in the compact form
\be \label{eq:Y3-Y4}
\begin{cases}
i\partial_t Y_3 = h Y_3- Y_3 h + k Y_4^* - Y_4 k^*, \\
i\partial_t Y_4 =  h Y_4 + Y_4 h^{\rm T} - k Y_3^{\rm T} + Y_3 k,\\
Y_3(0)=0, \quad Y_4(0)=0,
\end{cases}
\ee
where 
$$
Y_3:=\gamma+ \gamma^2 -\alpha \alpha^*, \quad Y_4:= \gamma \alpha - \alpha \gamma^{\rm T}.
$$
Here the initial conditions $Y_3(0)=0$ and $Y_4(0)=0$ follow from Lemma \ref{lem:dm} and the assumption that $\Phi(0)$ is a quasi-free state. 

The system \eqref{eq:Y3-Y4} is similar to \eqref{eq:linear-eq-homo} we have considered in Step 3, and by the same argument in the previous step we can show that $Y_3(t)=0$ and $Y_4(t)=0$ for all $t\in [0,1]$. Thus $\Phi(t)$ is a quasi-free state for all $t\in [0,1]$. Finally, since $\alpha\alpha^*=\gamma(1+\gamma)$ we obtain $\langle \Phi(t), \cN \Phi(t) \rangle = \Tr(\gamma(t)) \le \|\alpha\|_{\rm HS}^2$. Therefore, \eqref{eq:bound-cN-abstract} follows immediately from \eqref{eq:alpha-HS-gamma-final} and the simple estimate 
\begin{align*}\|\alpha(0)\|_{\rm HS}^2 + \|\gamma(0)\|_{\rm HS}^2 &= \Tr(\gamma(0)+ 2\gamma^2(0)) \\
&\le 2(1+\Tr(\gamma(0)))^2 = 2(1+\langle \Phi(0), \cN \Phi(0) \rangle)^2.\\[-2em]
\end{align*}
\end{proof}

\subsection{Proof of Proposition \ref{lem:evo-quasi-free}}

Now we apply Proposition \ref{lem:evol-quad-abstract} to the Bogoliubov Hamiltonian $\bH(t)$ defined by~\eqref{eq:Bogoliubov-Hamiltonian-Ht}, which corresponds to 
$$
h(t)= -\Delta+|u(t)|^2\ast w_N -\mu_N(t) + K_1(t), \quad k(t)=K_2(t).
$$
Recall that 
\begin{itemize}
\item $K_1(t)=Q(t)\tilde K_1(t)Q(t):\gH \to \gH$ where $\tilde{K}_1(t)$ is the operator on $\gH$ with kernel $\tilde{K}_1(t,x,y)=u(t,x)w_N(x-y)\overline{u(t,y)}$;

\item $K_2(t)=Q(t)\widetilde K_2(t) \overline{Q(t)}:\overline{\gH}\to \gH$ where $\widetilde K_2(t):\overline{\gH}\to \gH$ is the operator with kernel $\widetilde{K}_2(t,x,y)=u(t,x)w(x-y)u(t,y)$. Putting differently,\linebreak $K_2(t,\cdot,\cdot)=Q(t)\otimes Q(t)\widetilde{K}_2(t,\cdot,\cdot) \in \gH^2$.
\end{itemize}
We will decompose $h(t)=h_1+h_2(t)$ and $K_2=k_1+k_2$ with 
$$ h_1:=-\Delta,\quad h_2(t):=|u(t)|^2\ast w_N -\mu_N(t) + K_1(t)$$
and
$$ k_1 := \widetilde K_2, \quad k_2:= K_2- \widetilde K_2.$$
The conditions in Proposition \ref{lem:evol-quad-abstract} are verified in the following

\begin{lemma}\label{lem:vc} The following bounds hold true for all $\beta\ge 0$, $N\in \mathbb{N}$ and $t\ge 0$: 
\begin{align}
&\|h_2(t)\| + \|K_2\|  \le C \|u(t)\|_{L^\infty(\R^3)}^2, \label{eq:norm-h2-k}\\
&\|\partial_t h_2(t) \|  + \|K_2(t,\cdot,\cdot)\|_{L^2}+ \|\partial_t K_2(t,\cdot,\cdot)\|_{L^2} \le CN^{3\beta}, \label{eq:HS-K2} \\
& \|K_2(t,\cdot,\cdot)-\widetilde K_2 (t,\cdot,\cdot)\|_{L^2} \le C \|u(t)\|_{L^\infty(\R^3)}^2,
\label{eq:K2-tideK2} \\
\label{eq:HS-L1-K2}
&\|(-\Delta_x-\Delta_y)^{-1} \widetilde K_2(t, \cdot, \cdot) \|_{L^2} \\
&\qquad\qquad + \|(-\Delta_x-\Delta_y)^{-1} \partial_t \widetilde K_2 (t, \cdot, \cdot) \|_{L^2} \le C \|u(t)\|_{L^\infty(\R^3)}^{2/3}.\nn 
\end{align}
Here the constant $C$ depends only on $\|u(0)\|_{H^2(\R^3)}$.
\end{lemma}
\begin{proof} The bound \eqref{eq:norm-h2-k} follows from \eqref{eq:bound-wN*uu}-\eqref{eq:bound-muN}-\eqref{eq:bound-K1} and \eqref{eq:bound-K2-norm}. Next, we consider \eqref{eq:HS-K2}. Recall that we have proved $\|K_2(t,\cdot,\cdot)\|_{L^2}\le CN^{3\beta}$ in \eqref{eq:K2-L2}. Moreover, 
\begin{align} \label{eq:decomp-dt-K2} \partial_t K_2(t,\cdot,\cdot) &= (\partial_t Q(t))\otimes Q(t) \widetilde K_2(t,\cdot,\cdot)\\
&\quad+ Q(t)\otimes (\partial_t Q(t)) \widetilde K_2(t,\cdot,\cdot) \nn\\
&\quad + Q(t)\otimes Q(t) (\partial_t \widetilde K_2(t,\cdot,\cdot)). \nn
\end{align}
Note that $\partial_t Q(t)= - |\partial_t u(t) \rangle \langle u(t)| -  | u(t) \rangle \langle \partial_t u(t)|$ is a rank-two operator and $\|\partial_t Q(t)\|_{\rm HS}\le C$ because 
\be \label{eq:bound-dt-u}
\|\partial_t u(t)\|_{L^2} = \|h(t) u(t)\|_{L^2} \le \| \Delta u(t)\|_{L^2}+ \|h_2(t)u(t)\|_{L^2}\le C.
\ee
Consequently,
\begin{align} \label{eq:dt-Q-K2}
&\quad\ \left\| (\partial_t Q(t))\otimes Q(t) \widetilde K_2(t,\cdot,\cdot) \right\|_{L^2} + \left\| Q(t)\otimes (\partial_t Q(t)) \widetilde K_2(t,\cdot,\cdot) \right\|_{L^2} \\
&= \left\| (\partial_t Q(t))\widetilde K_2(t) \overline{Q(t)} \right\|_{\rm HS} + \left\| Q(t) \widetilde K_2(t)( \overline{\partial_t Q(t)}) \right\|_{\rm HS} \nn\\
&\le 2\|\partial_t Q(t)\|_{\rm HS} \|\widetilde K_2(t)\| \cdot \|Q(t)\| \le \|\widetilde K_2(t)\| \le C \|u(t)\|^2_{L^\infty(\R^3)} \nn.
\end{align}
On the other hand,
\begin{align} \label{eq:dt-K2-HS}
&\quad\ \| Q(t)\otimes Q(t) (\partial_t \widetilde K_2(t,\cdot,\cdot)) \|_{L^2}^2 \le \| \partial_t \widetilde K_2(t,\cdot,\cdot)\|_{L^2}^2 \\
&\le 2 \iint |\partial_t u(t,x)|^2 |w_N(x-y)|^2 |u(t,y)|^2 \mathrm{d} x \mathrm{d} y \nn\\
& \le 2\|\partial_t u(t)\|_{L^2(\R^3)}^2 \|w_N\|_{L^2(\R^3)}^2 \|u(t)\|_{L^\infty(\R^3)}^2 \le C N^{3\beta}.\nn
\end{align}
By the triangle inequality we deduce from \eqref{eq:decomp-dt-K2}, \eqref{eq:dt-Q-K2} and \eqref{eq:dt-K2-HS} that $
\|\partial_t K_2(t,\cdot,\cdot)\|_{L^2} \le C N^{3\beta}$. 
Similarly, we also have $\|\partial_t K_1(t)\|_{\rm HS}\le CN^{3\beta}$. Combining the latter inequality with
\begin{align*}
\left|\partial_t (|u(t)|^2*w_N)(x) \right| &= \left| 2\Re \int \overline{(\partial_t u)(t,y)}u(t,y) w_N(x-y) \mathrm{d} y \right| \\
&\le 2 \|\partial_t u(t)\|_{L^2(\R^3)} \|u(t)\|_{L^2(\R^3)} \|w_N\|_{L^\infty(\R^3)} \le C N^{3\beta}
\end{align*} 
we find that $\|\partial h_2(t)\|\le CN^{3\beta}$. Thus \eqref{eq:HS-K2} holds true.

The bound \eqref{eq:K2-tideK2} can be proved using similarly as \eqref{eq:dt-Q-K2}. More precisely, because $1-Q=|u(t)\rangle \langle u(t)|$, we have 
\begin{align*}
&\quad\ \| K_2(t,\cdot,\cdot) - \widetilde K_2(t,\cdot,\cdot) \|_{L^2} = \| Q(t) \widetilde K_2(t) \overline{Q(t)}- \widetilde K_2(t)\|_{\rm HS} \\
&\le  \| (Q(t)-1) \widetilde K_2(t) \overline{Q(t)} \|_{\rm HS} + \| \widetilde K_2(t)  (\overline{Q(t)}-1)\|_{\rm HS} \\
& \le \|Q(t)-1\|_{\rm HS} \|\widetilde K_2(t)\| \cdot \|\overline{Q(t)}\| + \|\widetilde K_2(t)\| \cdot \|\overline{Q(t)}-1\|_{\rm HS} \\
&\le C \|\widetilde K_2(t)\| \le C \|u(t)\|_{L^\infty(\R^3)}^2.
\end{align*}

The bound \eqref{eq:HS-L1-K2} is essentially \cite[Lemma 4.3]{GriMac-13}, but there is a technical modification that we will  clarify below. It suffices to consider the most complicated term $f(t,x,y):=u(t,x) w_N(x-y) (\partial_t u)(t,y)$ which is derived from  $\partial_t \widetilde K_2(t,x,y)$. Following \cite{GriMac-13}, we compute the Fourier transform :
\begin{align*}
\widehat{f}(t,p,q) &= \iint u(t,x)w_N(x-y)(\partial_t u)(t,y) e^{-i (p\cdot x + q\cdot y)} \mathrm{d} x \mathrm{d} y \\
&=\iint u(t,y+z) w_N(z) (\partial_t u)(t,y) e^{-i (p\cdot (y+z) + q\cdot y)} \mathrm{d} z\mathrm{d} y \\
&=\int w_N(z) \widehat{(u_z \partial_t u)}(t,p+q) e^{-ip\cdot z}\mathrm{d} z
\end{align*}
where $u_z(t,y):=u(t,y+z).$ Using the Cauchy-Schwarz inequality 
\begin{align*}
\left|\widehat{f}(t,p,q) \right|^2 \le \|w\|_{L^1} \int |w_N(z)| \cdot| \widehat{(u_z \partial_t u)}(t,p+q)|^2 \mathrm{d} z
\end{align*}
and Plancherel's Theorem, we can estimate  
\begin{align*}
&\quad\ \|(-\Delta_x-\Delta_y)^{-1} f(t, \cdot, \cdot)\|^2_{L^2} \\
&= (2\pi)^6 \iint (|p|^2+|q|^2)^{-2}\left|\widehat{f}(t,p,q) \right|^2 \mathrm{d} p \mathrm{d} q \\
&\le C \iiint (|p|^2+|q|^2)^{-2} |w_N(z)| \cdot |\widehat{(u_z \partial_t u)}(t,p+q)|^2 \mathrm{d} p \mathrm{d} q \mathrm{d} z \\
& = C \iiint (|p-q|^2+|q|^2)^{-2} |w_N(z)| \cdot |\widehat{(u_z \partial_t u)}(t,p)|^2 \mathrm{d} p \mathrm{d} q \mathrm{d} z \\
&\le C \iiint (|p|^2+|q|^2)^{-2} |w_N(z)| \cdot |\widehat{(u_z \partial_t u)}(t,p)|^2 \mathrm{d} p \mathrm{d} q \mathrm{d} z \\
&=  C \iint |w_N(z)| \cdot |p|^{-1}  |\widehat{(u_z \partial_t u)}(t,p)|^2 \mathrm{d} p \mathrm{d} z.
\end{align*}
By the Hardy-Littlewood-Sobolev inequality we get
\begin{align*}
&\quad\ \int |p|^{-1}  |\widehat{(u_z \partial_t u)}(t,p)|^2 \mathrm{d} p  \le C \|u_z(t) \partial_t u(t)\|_{L^{3/2}}^2 \\
&\le \|u_z(t)\|_{L^\infty}^{4/3} \|u_z(t)\|_{L^2}^{2/3} \|\partial_t u\|_{L^2}^2 \le C \|u(t)\|_{L^\infty}^{4/3}
\end{align*}
and this gives the desired bound
$$\|(-\Delta_x-\Delta_y)^{-1} f(t, \cdot, \cdot)\|^2_{L^2} \le C \|u(t)\|_{L^\infty}^{4/3}.$$
{\bf Clarification.} In \cite[Lemma 4.3]{GriMac-13}, the authors used $\|u_z(t) \partial_t u(t)\|_{L^{3/2}} \le\linebreak \|u(t)\|_{L^3} \|\partial_t u\|_{L^3}$, which we have avoided because we want to consider the case $u(0)\in H^2(\R^3)$ for which $\partial_t u(t)$ may be not in $L^3$.    
\end{proof}

Now we are able to give

\begin{proof}[Proof of Proposition \ref{lem:evo-quasi-free}] Recall that we will apply Proposition \ref{lem:evol-quad-abstract} to $h(t)=h_1+h_2(t)$ and $k=K_2=k_1+k_2$ with 
$$ h_1:=-\Delta,\quad h_2(t):=|u(t)|^2\ast w_N -\mu_N(t) + K_1(t)$$
and
$$ k_1 := \widetilde K_2, \quad k_2:= K_2- \widetilde K_2.$$
\medskip

\noindent
{\bf General results with $u(0)\in H^2(\R^3)$}. Recall that if $u(0)\in H^2(\R^3)$, then we have the uniform bound 
$$\|u(t)\|_{L^\infty(\R^3)} \le C_{\rm Sobolev} \|u(t)\|_{H^2(\R^3)}\le C$$
for a constant $C$ depending only on $\|u(0)\|_{H^2(\R^3)}$. Therefore, all relevant conditions in Proposition \ref{lem:evol-quad-abstract} have been verified by Lemma \ref{lem:vc}. Thus by Proposition \ref{lem:evol-quad-abstract}, the Bogoliubov equation \eqref{eq:eq-Phit} has a unique global solution $\Phi(t)$ and the pair of density matrices $(\gamma(t),\alpha(t))=(\gamma_{\Phi(t)}, \alpha_{\Phi(t)})$ is the unique solution to \eqref{eq:linear-Bog-dm}. The constraint $\Phi(t)\in \cF_+(t)$ follows from the Hartree equation, see \cite[Proof of Theorem 7]{LewNamSch-14} for an explanation.

Now we consider the bound \eqref{eq:bound-alpha-HS}. Since 
$$ \xi(t):=6 (\|h_2(t)\| + \|K_2(t)\|) \le C$$
and
\begin{align*} \Theta(t) &:= 2\|\alpha(0)\|_{\rm HS}^2  + 2\|\gamma(0)\|_{\rm HS}^2  \\
&\quad + 2\bigg( \| (-\Delta_x -\Delta_y)^{-1} \widetilde K_2(t,\cdot, \cdot)\|_{L^2} + \| (-\Delta_x -\Delta_y)^{-1} \widetilde K_2 (0,\cdot, \cdot)\|_{L^2}  \\
&\quad +\int_0^t \big( \|(K_2- \widetilde K_2)(s,\cdot,\cdot)\|_{L^2}+ \| (-\Delta_x -\Delta_y)^{-1} \partial_s \widetilde K_2 (s,\cdot, \cdot)\|_{L^2} \big) \mathrm{d} s \bigg)^2\\
&\le 2\|\alpha(0)\|_{\rm HS}^2  + 2\|\gamma(0)\|_{\rm HS}^2 + C (1+t^2)
\end{align*}
we deduce from \eqref{eq:bound-alpha-HS} that
\begin{align*}
\|\alpha(t)\|_{\rm HS}^2+ \|\gamma(t)\|_{\rm HS}^2 &\le \Theta(t) + \int_0^t \exp\Big( \int_s^t \xi(r) \mathrm{d} r \Big) \xi(s) \Theta (s) \mathrm{d} s \\
&\le e^{Ct} ( 1 + \|\alpha(0)\|_{\rm HS}^2 + \|\gamma(0)\|_{\rm HS}^2 )
\end{align*}
for a constant $C$ depending only on $\|u(0)\|_{H^2(\R^3)}$. 

Moreover, by Proposition \ref{lem:evol-quad-abstract}, if $\Phi(0)$ is a quasi-free state, then $\Phi(t)$ is a quasi-free state for all $t\ge 0$, and from \eqref{eq:bound-cN-abstract} we obtain
$$
\langle \Phi(t), \cN \Phi(t) \rangle \le e^{Ct} \Big( 1 + \langle \Phi(0), \cN \Phi(0)\rangle \Big)^2. 
$$
\smallskip

\noindent
{\bf Improved bound with $u(0)$ smooth}. Now we consider the case when the initial Hartree state is smooth, namely $u(0)\in W^{\ell,1}(\R^3)$ with $\ell$ sufficiently large. Recall that in this case 
$$ \|u(t)\|_{L^\infty(\R^3)}\le \frac{C_1}{(1+t)^{3/2}}$$
for a constant $C_1$ depending only on $\|u(0)\|_{W^{\ell,1}(\R^3)}$. By Proposition \ref{lem:evol-quad-abstract}, we have 
$$ \xi(t)=6 (\|h_2(t)\| + \|K_2(t)\|) \le \frac{C_1}{(1+t)^3}$$
and
\begin{align*} \Theta(t) &= 2\|\alpha(0)\|_{\rm HS}^2  + 2\|\gamma(0)\|_{\rm HS}^2  \\
&\quad+ 2\bigg( \| (-\Delta_x -\Delta_y)^{-1} \widetilde K_2(t,\cdot, \cdot)\|_{L^2} + \| (-\Delta_x -\Delta_y)^{-1} \widetilde K_2 (0,\cdot, \cdot)\|_{L^2}  \\
&\quad+\int_0^t \big( \|(K_2- \widetilde K_2)(s,\cdot,\cdot)\|_{L^2}+ \| (-\Delta_x -\Delta_y)^{-1} \partial_s \widetilde K_2 (s,\cdot, \cdot)\|_{L^2} \big) \mathrm{d} s \bigg)^2,\\
&\le 2\|\alpha(0)\|_{\rm HS}^2  + 2\|\gamma(0)\|_{\rm HS}^2 + C_1 \left( \frac{1}{1+t} + \int_0^t \Big( \frac{1}{(1+t)^3} + \frac{1}{1+t} \Big) \mathrm{d} s \right)^2 \\
& \le 2\|\alpha(0)\|_{\rm HS}^2  + 2\|\gamma(0)\|_{\rm HS}^2  + C_1 \log^2 (1+t).
\end{align*}
Since $(1+t)^{-3}$ in integrable, the estimate \eqref{eq:bound-alpha-HS} in Proposition \ref{lem:evol-quad-abstract} gives us
\begin{align*}
\|\alpha(t)\|_{\rm HS}^2+ \|\gamma(t)\|_{\rm HS}^2 &\le \Theta(t) + \int_0^t \exp\Big( \int_s^t \xi(r) \mathrm{d} r \Big) \xi(s) \Theta (s) \mathrm{d} s \\
&\le C_1 ( \log^2(1+t) + \|\alpha(0)\|_{\rm HS}^2 + \|\gamma(0)\|_{\rm HS}^2 ).
\end{align*}
Moreover, if $\Phi(0)$ is a quasi-free state, then from \eqref{eq:bound-cN-abstract} we obtain
$$
\langle \Phi(t), \cN \Phi(t) \rangle \le C_1 \Big( \log(1+t) + 1 + \langle \Phi(0), \cN \Phi(0)\rangle \Big)^2
$$
for a constant $C_1$ depending only on $\|u(0)\|_{W^{\ell,1}(\R^3)}$.
\end{proof}

\section{Proof of Main Theorem} \label{sec:proof-main-thm}

Assuming Lemma \ref{lem:Nk-quasi-free} at the moment, we are ready to provide

\begin{proof}[Proof of Theorem \ref{thm:main}] We will compare $\Phi_N(t)=U_N(t)\Psi_N(t)$ with the Bogoliubov evolution $\Phi(t)$. Using the equations \eqref{eq:eq-PhiNt} and \eqref{eq:eq-Phit}, we can compute 
\begin{align} \label{eq:dt-PhiNt-Phit-norm}
&\quad\ \partial_t \| \Phi_{N}(t)-\Phi(t)\|^2 \\
&= 2 {\rm Im} \Big\langle i \partial_t \Phi_N(t), \Phi(t) \big\rangle - 2 {\rm Im} \Big\langle \Phi_N(t), i\partial_t \Phi(t) \Big\rangle \nn\\ 
&= 2 {\rm Im} \Big \langle (\1_+^{\le N} \bH(t)\1_+^{\le N} + R(t)) \Phi_N(t), \Phi(t) \Big\rangle -  2 {\rm Im}  \Big\langle \Phi_N (t),  \bH \Phi(t) \Big\rangle \nn \\
& = 2 {\rm Im} \Big\langle R(t) \Phi_N(t), \Phi(t) \Big\rangle - 2 {\rm Im} \Big\langle \Phi_N (t), \1_+^{\le N} \bH (1-\1_+^{\le N}) \Phi(t) \Big\rangle\nn
\end{align}
where $\1_+^{\le N}:=\1_{\cF_+^{\le N}(t)}$ and 
$$ R(t):= \1_+^{\le N} \Big[i \left( \partial_t U_N(t)\right) U_N^*(t) + U_N(t)H_N U_N^*(t) -  \bH(t) \Big] \1_+^{\le N}.$$
Using the Cauchy-Schwarz inequality and Proposition \ref{lem:Bog-app}, we can estimate
\begin{align*} 
&\quad\ {\rm Im} \langle R(t) \Phi_N(t), \Phi(t) \rangle \\
&= {\rm Im} \langle \Phi_N(t)-\Phi(t), R(t) \Phi(t) \rangle \nn\\
& \le \| R(t) \Phi(t)\| \cdot \|\Phi_N(t)- \Phi(t)\| \nn\\
&\le C N^{(3\beta-1)/2} \Big \langle \Phi(t), (\cN_+ +1)^4 \Phi(t) \Big\rangle^{1/2}\|\Phi_N(t)- \Phi(t)\|.
\end{align*}
Moreover, by the Cauchy-Schwarz inequality again and \eqref{eq:Kcr*-Kcr} we get
\begin{align*}
&\quad\ {\rm Im} \Big\langle \Phi_N (t), \1_+^{\le N} \bH (1-\1_+^{\le N}) \Phi(t) \Big\rangle\\
&\le \| \1_+^{\le N} \bH (1-\1_+^{\le N}) \Phi(t)  \|\\
&= \| \1_+^{\le N} \mathbb{K}_{\rm cr}^* (1-\1_+^{\le N}) \Phi(t)  \| \\
&\le \Big \langle \Phi(t), \mathbb{K}_{\rm cr}  \mathbb{K}_{\rm cr}^* (1-\1_+^{\le N}) \Phi(t)  \Big\rangle^{1/2} \\
&\le C N^{3\beta/2} \Big \langle \Phi(t), (\cN_+ + 1)^2 (1-\1_+^{\le N}) \Phi(t)  \Big\rangle^{1/2} \\
&\le CN^{3\beta/2-1} \Big \langle \Phi(t), (\cN_+ + 1)^4 \Phi(t)  \Big\rangle^{1/2} .
\end{align*}
Thus from \eqref{eq:dt-PhiNt-Phit-norm} it follows that 
\begin{align*}
\partial_t \| \Phi_{N}(t)-\Phi(t)\|^2 &\le CN^{(3\beta-1)/2} \Big\langle \Phi(t), (\cN_+ +1)^4 \Phi(t) \Big\rangle^{1/2} \\
&\quad \times \Big( \|\Phi_N(t)- \Phi(t)\| + N^{-1/2} \Big).
\end{align*}
Consequently, the function $f(t):= \| \Phi_{N}(t)-\Phi(t)\|^2+N^{-1}$ satisfies
$$
\partial_t f(t) \le CN^{(3\beta-1)/2} \Big\langle \Phi(t), (\cN_+ +1)^4 \Phi(t) \Big\rangle^{1/2} \sqrt{f(t)},
$$
and hence  
$$
\partial_t \sqrt{f(t)} \le CN^{(3\beta-1)/2} \Big\langle \Phi(t), (\cN_+ +1)^4 \Phi(t) \Big\rangle^{1/2}.
$$
Taking the integral over $t$, we obtain 
\begin{align}  \label{eq:dt-PhiNt-Phit-norm-ineq1}
&\quad\ \Big( \| \Phi_{N}(t)-\Phi(t)\|^2+ N^{-1} \Big)^{1/2} \\
& \le \Big(\| \Phi_{N}(0)-\Phi(0)\|^2+N^{-1}\Big)^{1/2}\nn  \\
&\quad + t CN^{(3\beta-1)/2} \sup_{s\in [0,t]}\Big\langle \Phi(s), (\cN_+ +1)^4 \Phi(s) \Big\rangle^{1/2}.\nn
\end{align} 
Now we make a further estimate for every term in \eqref{eq:dt-PhiNt-Phit-norm-ineq1}. First, since $U_N(t):\gH^N\to \cF_+^{\le N}(t)$ is a unitary operator, we have
\begin{align} \label{eq:dt-PhiNt-Phit-norm-ineq1a}
\| \Phi_{N}(t)-\Phi(t)\|^2 &= \| U_N(t) \Psi_N(t) - \1_+^{\le N} \Phi(t)\|^2 + \| (1- \1_+^{\le N}) \Phi(t) \|^2 \\
& \ge \|  \Psi_N(t) - U_N(t) \1_+^{\le N} \Phi(t)\|^2 .\nn
\end{align}
On the other hand, the choice $\Phi_N(0)= \1_+^{\le N} \Phi(0)$ implies that 
\begin{align} \label{eq:dt-PhiNt-Phit-norm-ineq1b}
\| \Phi_{N}(0)-\Phi(0)\|^2 = \Big\langle \Phi(0), (\1-\1_+^{\le N})\Phi(0) \Big\rangle \le \frac{1}{N}\Big\langle \Phi(0), \cN_+ \Phi(0) \Big\rangle.
\end{align}
Moreover, recall that $\Phi(t)$ is a quasi-free state for all $t$ due to Proposition \ref{lem:evo-quasi-free} and the assumption that $\Phi(0)$ is a quasi-free state. Therefore, by Lemma \ref{lem:Nk-quasi-free},
\begin{align} \label{eq:dt-PhiNt-Phit-norm-ineq1c}  \Big\langle \Phi(t), (\cN +1)^4 \Phi(t) \Big\rangle \le C \Big\langle \Phi(t), (\cN +1) \Phi(t) \Big\rangle^4.
\end{align}
Inserting \eqref{eq:dt-PhiNt-Phit-norm-ineq1a}-\eqref{eq:dt-PhiNt-Phit-norm-ineq1b}-\eqref{eq:dt-PhiNt-Phit-norm-ineq1c} into \eqref{eq:dt-PhiNt-Phit-norm-ineq1} we find that
\begin{align} \label{eq:semi-final-estimate}
&\quad\ \|  \Psi_N(t) - U_N(t) \1_+^{\le N} \Phi(t)\| \\
&\le \frac{1}{\sqrt{N}}\Big|\langle \Phi(0), (\cN +1) \Phi(0) \rangle \Big|^{1/2} \nn\\
&\quad + tCN^{(3\beta-1)/2} \sup_{s\in [0,t]} \Big\langle \Phi(s), (\cN +1) \Phi(s) \Big\rangle^2 .\nn
\end{align}
Finally, we derive \eqref{eq:PsiN-app-thm} from \eqref{eq:semi-final-estimate} and the the upper bounds on $\langle \Phi(t),\cN \Phi(t) \rangle$ in Proposition \ref{lem:evo-quasi-free}. To be precise, if we only know $u(0)\in H^2(\R^3)$, then from \eqref{eq:semi-final-estimate} and \eqref{eq:bound-N-Phit} we obtain
\begin{align}
&\|  \Psi_N(t) - U_N(t) \1_+^{\le N} \Phi(t)\| \le N^{(3\beta-1)/2} e^{Ct}  \Big( 1 + \langle \Phi(0), \cN \Phi(0) \rangle \Big)^4
\end{align}  
for a constant $C$ depending only on $\|u(0)\|_{H^2(\R^3)}$. If we know that $u(0)\in W^{\ell,1}(\R^3)$ with $\ell$ sufficiently large, then from \eqref{eq:semi-final-estimate} and \eqref{eq:bound-N-Phit-improved} we get the improved bound 
\begin{align}
&\quad\ \|  \Psi_N(t) - U_N(t) \1_+^{\le N} \Phi(t)\| \\
&\le N^{(3\beta-1)/2} C_1 (1+t) \Big( \log(1+t) + 1 + \langle \Phi(0), \cN \Phi(0)\rangle \Big)^4\nn
\end{align}  
for a constant $C_1$ depending only on $\|u(0)\|_{W^{\ell,1}(\R^3)}$. 
\end{proof}

For the completeness, let us provide

\begin{proof}[Proof of Lemma \ref{lem:Nk-quasi-free}] Since the density matrices $\gamma_\Psi,\alpha_\Psi$ of $\Psi$ satisfy the relations \eqref{eq:dm-quasi-free}, we can diagonalize them simultaneously. More precisely (see also \cite[Lemma 8]{NamNapSol-15}), we can find  an orthonormal basis  $\{u_n\}_{n=1}^\infty $ for $\gH$ and non-negative numbers $\{\lambda_n\}_{n=1}^\infty$ such that
$$ \gamma_\Psi = \sum_{n=1}^\infty \lambda_n |u_n\rangle \langle u_n|, \quad \alpha_\Psi u_n = \sqrt{\lambda_n+\lambda_n^2} (\overline{u_n}), \quad \forall n \in \mathbb{N}.$$
Let us denote $a_n=a(u_n)$ for short. By the definition of $\gamma_\Psi$ and $\alpha_\Psi$, we have 
\begin{align} \label{eq:am-an-lambdan}
\langle \Psi, a^*_m a_n \Psi \rangle = \delta_{m,n} \lambda_n, \quad  \langle \Psi, a_m a_n \Psi \rangle =  \delta_{m,n} \sqrt{\lambda_n+\lambda_n^2},\quad \forall n \in \mathbb{N}.
\end{align}
Moreover, 
$$\sum_n \lambda_n = \Tr(\gamma_\Phi)=\langle \Psi, \cN \Psi \rangle.$$ 
Now we compute
\begin{align*}
&\quad\ \langle \Psi,  \cN(\cN-1)(\cN-2)(\cN-3)(\cN-\ell+1) \Psi \rangle \\
& = \sum_{n_1,n_2,\dots, n_\ell \ge 1} \langle \Psi, a^*_{n_1} \cdots a^*_{n_\ell}  a_{n_1} \cdots a_{n_\ell} \Psi \rangle \\
& = \sum_{1\le s \le \ell} ~~~ \sum_{\substack{ 1\le n_1<n_2<\cdots<n_s \\ (m_1,\dots,m_s)\in P(s,\ell)}} \langle \Psi, (a^*_{n_1})^{m_1} (a_{n_1})^{m_1} \cdots (a^*_{n_s})^{m_s}  (a_{n_s})^{m_s} \Psi \rangle,
\end{align*}
where $P(s,\ell)$ is the set of partitions:
$$
P(s,\ell)= \{ (m_1,\dots,m_s)\subset \mathbb{N} ~|~ m_1+\cdots +m_s=\ell\}.
$$
For all $1\le n_1<n_2<\cdots<n_s$ and $m_1,\dots,m_s \ge 1$, by using Wick's Theorem and \eqref{eq:am-an-lambdan} we find that
\begin{align*}
& \langle \Psi, (a^*_{n_1})^{m_1} (a_{n_1})^{m_1} \cdots (a^*_{n_s})^{m_s}  (a_{n_s})^{m_s} \Psi \rangle = \prod_{j=1}^{s}\langle \Psi, (a^*_{n_j})^{m_j} (a_{n_j})^{m_j} \Psi \rangle 
\end{align*}
and
\begin{equation*}
\langle \Psi, (a^*_{n_j})^{m_j} (a_{n_j})^{m_j} \Psi \rangle \le |P(2m_j)| \lambda_{n_j} (1+\lambda_{n_j})^{m_j-1}. 
\end{equation*}
Here recall that
$$
P_{2n} = \{\sigma \in S(2n)~ |~ \sigma(2j-1)<\min\{\sigma(2j),\sigma(2j+1)\}~\text{for~all}~ j\}.
$$
is the set of paring and we have denoted by $|P(2n)|$ the number of elements of $P_{2n}$. Thus
\begin{align*}
& \quad \sum_{\substack{ 1\le n_1<n_2<\cdots<n_s \\ (m_1,\dots,m_s)\in P(s,\ell)}} \langle \Psi, (a^*_{n_1})^{m_1} (a_{n_1})^{m_1} \cdots (a^*_{n_s})^{m_s}  (a_{n_s})^{m_s} \Psi \rangle \\
& \le   \sum_{\substack{ 1\le n_1<n_2<\cdots<n_s \\ (m_1,\dots,m_s)\in P(s,\ell)}} \prod_{j=1}^{s} |P(2m_j)| \lambda_{n_j} (1+\lambda_{n_j})^{m_j-1} \\
& \le \sum_{\substack{ 1\le n_1<n_2<\cdots<n_s \\ (m_1,\dots,m_s)\in P(s,\ell)}} |P(2\ell)| (1+ \sup_p \lambda_{p})^{ \ell -s}  \prod_{j=1}^{s} \lambda_{n_j} 
\end{align*}
\begin{align*}
& \le |P(s,\ell)| \cdot |P(2\ell)| \cdot \left(1+ \sup_p \lambda_{p} \right)^{ \ell -s}   \sum_{ 1\le n_1<n_2<\cdots<n_s} \prod_{j=1}^{s}  \lambda_{n_j} \\
& \le |P(s,\ell)| \cdot |P(2\ell)| \cdot \left(1+ \sup_p \lambda_{p} \right)^{ \ell -s}  \left(\sum_p \lambda_p \right)^s \\
& \le |P(s,\ell)| \cdot |P(2\ell)| \cdot \left(1+ \sum_p \lambda_{p} \right)^{\ell} \\
& = |P(s,\ell)| \cdot |P(2\ell)| \cdot \left(1+ \langle \Psi,  \cN \Psi \rangle \right)^{\ell} 
\end{align*}
Here in the second estimate we have used $\prod_{j=1}^{s} |P(2m_j)| \le |P(2\ell)|$, which follows obviously from the definition of $P(2n)$ and the fact $\sum_{j=1}^s m_j=\ell$. Summing the latter estimate over $s=1,2,\dots,\ell$ we obtain
\begin{align*}
&\quad\ \langle \Psi,  \cN(\cN-1)(\cN-2)(\cN-3)(\cN-\ell+1) \Psi \rangle \\
& = \sum_{1\le s \le \ell} ~~~ \sum_{\substack{ 1\le n_1<n_2<\cdots<n_s \\ (m_1,\dots,m_s)\in P(s,\ell)}} \langle \Psi, (a^*_{n_1})^{m_1} (a_{n_1})^{m_1} \cdots (a^*_{n_s})^{m_s}  (a_{n_s})^{m_s} \Psi \rangle \\
& \le \left(\sum_{s=1}^\ell |P(s,\ell)| \right) \cdot |P(2\ell)| \cdot \left(1+ \langle \Psi,  \cN \Psi \rangle \right)^{\ell} .
\end{align*}
The desired result follows from a straightforward induction argument.
\end{proof}

\noindent{\bf Acknowledgments.} We thank Jan Derezi\'nski, Mathieu Lewin, Nicolas Rougerie and Robert Seiringer for interesting discussions and useful remarks. We gratefully acknowledge the financial supports by the REA grant agreement No. 291734 (PTN), as well as the support of the National Science Center (NCN) grant No. 2012/07/N/ST1/03185 and the Austrian Science Fund (FWF) project Nr. P 27533-N27 (MN).


\begin{thebibliography}{99}

\bibitem{AdaGolTet-07}
{R.~Adami, F.~Golse, and A.~Teta}, 
{\em Rigorous derivation of the cubic NLS
in dimension one}, J.~Stat. Phys. {\bf 127} (6) (2007), 1193--1220.

\bibitem{AmmNie-08}
{Z.~Ammari and F.~Nier}, 
{\em Mean field limit for bosons and infinite
  dimensional phase-space analysis}, Ann. Henri Poincar\'e {\bf 9} (2008),
  1503--1574.

\bibitem{BacBreKonMen-14}
{V.~Bach, S.~Breteaux, H.~K. K\"onrr, and E.~Menge}, {\em Generalization of
  {L}ieb's variational principle to {B}ogoliubov-{H}artree-{F}ock theory}, 
  J.~Math. Phys. {\bf 55} (2014), 012101.
  
\bibitem{BarGolMau-00}
{C.~Bardos, F.~Golse, and N.J.~Mauser}, 
{\em Weak coupling limit of the N-particle Schr\"odinger equation}, 
Methods Appl. Anal. {\bf 7} (2000), no.~2, 275--293.  

\bibitem{BenOliSch-15}
{N.~{Benedikter}, G.~{de Oliveira}, and B.~{Schlein}}, {\em Quantitative
  derivation of the Gross-Pitaevskii equation}, 
  Comm. Pure App. Math. {\bf 68}
  (2015), no.~8, 1399--1482.

\bibitem{BenLie-83}
{R.~{Benguria} and E.~H. {Lieb}}, {\em Proof of the stability of highly
  negative ions in the absence of the Pauli principle}, 
  Phys. Rev. Lett. {\bf 50}  (1983), 1771--1774.

\bibitem{Berezin-66}
{F.~Berezin}, {\em The method of second quantization}, Pure and applied
  physics, A series of monographs and textbooks, Academic Press, 1966.

\bibitem{BocCenSch-15}
{C.~{Boccato}, S.~{Cenatiempo}, and B.~{Schlein}}, {\em Quantum 
many-body fluctuations around nonlinear {S}chr\"odinger dynamics}, 
Ann. Henri Poincar\'e {\bf 18} (2017), no.~1, 113--191. 


\bibitem{Bogoliubov-47b}
{N.~N. Bogoliubov}, {\em On the theory of superfluidity}, J.~Phys. (USSR),
  {\bf 11} (1947), 23.

\bibitem{Bourgain-98}
{J.~Bourgain}, {\em Scattering in the energy space and below for 3{D}
  {NLS}}, J.~d'Anal. Math. {\bf 75} (1998), 267--297.

\bibitem{CheLeeSch-11}
{L.~{Chen}, J.~O. {Lee}, and B.~{Schlein}}, {\em {Rate of convergence
  towards Hartree dynamics}}, J.~Stat. Phys. {\bf 144} (2011), 872--903.

\bibitem{CheHaiPavSei-15}
{T.~{Chen}, C.~{Hainzl}, N.~{Pavlovic}, and R.~{Seiringer}}, 
{\em Unconditional uni\-queness for the cubic Gross-Pitaevskii hierarchy via
  quantum de Finetti}, Commun. Pure Appl. Math. {\bf 68} (2015), no.~10,
  1845--1884.

\bibitem{CheHol-13}
{X.~Chen and J.~Holmer}, {\em Focusing quantum many-body dynamics: The rigorous derivation of the 1D focusing cubic nonlinear Schr\"odinger equation}, Arch. Rational Mech. Anal. {\bf 221} (2016), 631--676.

  
\bibitem{CheHol-15}
{X.~Chen and J.~Holmer}, {\em The rigorous derivation of the 2D cubic focusing {NLS} from quantum many-body evolution}, Int. Math. Res. Not. (2017), no.~14,  4173--4216

\enlargethispage{1em}
\bibitem{DerNap-13}
{J.~{Derezi{\'n}ski} and M.~{Napi{\'o}rkowski}}, {\em Excitation spectrum
  of interacting bosons in the mean-field infinite-volume limit}, Annales
  Henri Poincar\'e {\bf 15} (2014), 2409--2439.
 Erratum: Annales Henri Poincar\'e {\bf 16} (2015), 1709--1711.

\bibitem{ErdSchYau-07}
{L.~Erd{\"{o}}s, B.~Schlein, and H.-T. Yau}, {\em Derivation of the cubic
  non-linear {S}chr\"odinger equation from quantum dynamics of many-body
  systems}, Invent. Math. {\bf 167} (2007), 515--614.

\bibitem{ErdSchYau-09}
{L.~Erd{\"{o}}s, B.~Schlein, and H.-T. Yau}, {\em Rigorous derivation
  of the {G}ross-{P}itaevskii equation with a large interaction potential}, 
  J.~Amer. Math. Soc. {\bf 22} (2009), 1099--1156.

\bibitem{ErdSchYau-10}
{L.~Erd{\"{o}}s, B.~Schlein, and H.-T. Yau}, {\em Derivation of the
  {G}ross-{P}itaevskii equation for the dynamics of {B}ose-{E}instein
  condensate}, Ann. of Math. (2) {\bf 172} (2010), 291--370.

\bibitem{ErdYau-01}
{L.~Erd{\"o}s and H.-T. Yau}, {\em Derivation of the nonlinear
  {S}chr\"odinger equation from a many body {C}oulomb system}, Adv. Theor.
  Math. Phys. {\bf 5} (2001), 1169--1205.

\bibitem{FanSpoVer-80}
{M.~Fannes, H.~Spohn, and A.~Verbeure}, {\em Equilibrium states for mean
  field models}, J.~Math. Phys. {\bf 21} (1980), 355--358.

\bibitem{FroKnoSch-09}
{J.~Fr{\"o}hlich, A.~Knowles, and S.~Schwarz}, {\em On the mean-field limit
  of bosons with {C}oulomb two-body interaction}, Commun. Math. Phys. {\bf 288}
  (2009), 1023--1059.

\bibitem{GinVel-79}
{J.~Ginibre and G.~Velo}, {\em The classical field limit of scattering
  theory for nonrelativistic many-boson systems. {I}}, Commun. Math. Phys. {\bf 66}
  (1979), 37--76.

\bibitem{GinVel-79b}
{J.~Ginibre and G.~Velo}, {\em The classical field
  limit of scattering theory for nonrelativistic many-boson systems. {II}},
  Commun. Math. Phys. {\bf 68} (1979), 45--68.

\bibitem{GreSei-13}
{P.~Grech and R.~Seiringer}, {\em The excitation spectrum for weakly
  interacting bosons in a trap}, Comm. Math. Phys. {\bf 322} (2013), 559--591.

\bibitem{GriMac-13}
{M.~{Grillakis} and M.~{Machedon}}, {\em Pair excitations and the mean
  field approximation of interacting {B}osons, I}, Commun. Math. Phys. {\bf 324}
  (2013), 601--636.
   
  
\bibitem{GriMac-15}
{M.~{Grillakis} and M.~{Machedon}}, {\em Pair excitations and the mean
  field approximation of interacting {B}osons, II}, Comm. PDE. {\bf 42} (2017), no.~1, 24--67.


\bibitem{GriMacMar-10}
{M.~G. Grillakis, M.~Machedon, and D.~Margetis}, {\em Second-order
  corrections to mean field evolution of weakly interacting bosons. {I}},
  Commun. Math. Phys. {\bf 294} (2010), 273--301.

\bibitem{GriMacMar-11}
{M.~G. Grillakis, M.~Machedon, and D.~Margetis}, {\em Second-order
  corrections to mean field evolution of weakly interacting bosons. {II}}, Adv.
  Math. {\bf 228} (2011), 1788--1815.

\bibitem{Hepp-74}
{K.~Hepp}, {\em The classical limit for quantum mechanical correlation
  functions}, Comm. Math. Phys. {\bf 35} (1974), 265--277.

\bibitem{KlaMac-08}
{S.~Klainerman and M.~Machedon}, {\em On the uniqueness of solutions to the
  gross-pitaevskii hierarchy}, Commun. Math. Phys. {\bf 279} (2008), 169--185.

\bibitem{KnoPic-10}
{A.~Knowles and P.~Pickl}, {\em Mean-field dynamics: singular potentials
  and rate of convergence}, Commun. Math. Phys. {\bf 298} (2010), 101--138.
  

 \bibitem{Kuz-15b}
{E.~Kuz}, {\em Exact Evolution versus Mean Field with Second-order correction 
for Bosons Interacting via Short-range Two-body Potential}, Differential Integral Equations {\bf 30} (2017), no.~7/8, 587--630.

\bibitem{LewNamRou-14d}
{M.~Lewin, P.~Nam, and N.~Rougerie}, {\em Derivation of nonlinear {G}ibbs
  measures from many-body quantum mechanics}, Journal de l' \'Ecole
  polytechnique -- Math\'ematiques {\bf 2} (2015), 65--115.

\bibitem{LewNamRou-14}
{M.~Lewin, P.~T. Nam, and N.~Rougerie}, {\em Derivation of {H}artree's
  theory for generic mean-field {B}ose gases}, Adv. Math. {\bf 254} (2014),
  570--621.

\bibitem{LewNamRou-14c}
{M.~Lewin, P.~T. Nam, and N.~Rougerie}, {\em The mean-field
  approximation and the non-linear {S}chr\"odinger functional for trapped
  {B}ose gases}, Trans. Amer. Math. Soc. {\bf 368} (2016), 6131--6157. 
  

\bibitem{LewNamSch-14}
{M.~Lewin, P.~T. Nam, and B.~Schlein}, {\em Fluctuations around {H}artree
  states in the mean-field regime}, Amer. J. Math. {\bf 137} (2015), no.~6, 1613--1650. 

\bibitem{LewNamSerSol-15}
{M.~Lewin, P.~T. Nam, S.~Serfaty, and J.~P. Solovej}, {\em Bogoliubov
  spectrum of interacting {B}ose gases}, Comm. Pure Appl. Math. {\bf 68} (2015),
  413--471.

\bibitem{LieLin-63}
{E.~H. Lieb and W.~Liniger}, {\em Exact analysis of an interacting {B}ose
  gas. {I}. {T}he general solution and the ground state}, Phys. Rev. (2) {\bf 130}
  (1963), 1605--1616.

\bibitem{LieSei-02}
{E.~H. Lieb and R.~Seiringer}, {\em {Proof of Bose-Einstein condensation
  for dilute trapped gases}}, Phys. Rev. Lett., 88 (2002), 170409.

\bibitem{LieSei-06}
{E.~H. Lieb and R.~Seiringer}, {\em Derivation of the
  {G}ross-{P}itaevskii equation for rotating {B}ose gases}, Commun. Math.
  Phys. {\bf 264} (2006), 505--537.

\enlargethispage{1em}
\bibitem{LieSeiYng-00}
{E.~H. Lieb, R.~Seiringer, and J.~Yngvason}, {\em Bosons in a trap: A
  rigorous derivation of the {G}ross-{P}itaevskii energy functional}, Phys.
  Rev. A {\bf 61} (2000), 043602.

\bibitem{LieYau-87}
{E.~H. Lieb and H.-T. Yau}, {\em The {C}handrasekhar theory of stellar
  collapse as the limit of quantum mechanics}, Commun. Math. Phys. {\bf 112} (1987),
  147--174.

\bibitem{NamNapSol-15}
{P.~T. Nam, M.~Napi\'orkowski, and J.~P. Solovej}, 
{\em Diagonalization of bosonic quadratic {H}amiltonians by {B}ogoliubov transformations}, 
J. Funct. Anal. {\bf 270} (2016), no.~11, 4340--4368.



\bibitem{NamRouSei-15}
{P.~Nam, N.~Rougerie, and R.~Seiringer}, {\em Ground states of large {B}ose
  systems: {T}he {G}ross-{P}itaevskii limit revisited}, Anal. \& PDE. {\bf 9} (2016), no.~2, 459--485.
  
  

\bibitem{NamSei-15}
{P.~Nam and R.~Seiringer}, {\em Collective excitations of {B}ose gases in
  the mean-field regime}, Arch. Rational Mech. Anal. {\bf 215} (2015), 381--417.

\bibitem{Nam-thesis}
{P.~T. Nam}, {\em Contributions to the rigorous study of the structure of
  atoms}, PhD thesis, University of Copenhagen, 2011.

\bibitem{PetRagVer-89}
{D.~Petz, G.~A. Raggio, and A.~Verbeure}, {\em Asymptotics of
  {V}aradhan-type and the {G}ibbs variational principle}, 
  Comm. Math. Phys. {\bf 121} (1989), 271--282.

\bibitem{Pickl-11}
{P.~Pickl}, {\em A simple derivation of mean-field limits for quantum
  systems}, Lett. Math. Phys. {\bf 97} (2011), 151--164.

\bibitem{Pickl-15}
{P.~Pickl}, {\em Derivation of the
  time dependent {G}ross {P}itaevskii equation with external fields}, Rev.
  Math. Phys. {\bf 27} (2015), 1550003.

\bibitem{RagWer-89}
{G.~A. Raggio and R.~F. Werner}, {\em Quantum statistical mechanics of
  general mean field systems}, Helv. Phys. Acta {\bf 62} (1989), 980--1003.

\bibitem{ReeSim2}
{M.~Reed and B.~Simon}, {\em Methods of {m}odern {m}athematical {p}hysics.
  {II}. {F}ourier analysis, self-adjointness}, Academic Press, New York, 1975.

\bibitem{RodSch-09}
{I.~Rodnianski and B.~Schlein}, {\em Quantum fluctuations and rate of
  convergence towards mean field dynamics}, Commun. Math. Phys. {\bf 291} (2009),
  31--61.

\bibitem{Seiringer-11}
{R.~Seiringer}, {\em The excitation spectrum for weakly interacting
  bosons}, Commun. Math. Phys. {\bf 306} (2011), 565--578.

\bibitem{SeiYngZag-12}
{R.~Seiringer, J.~Yngvason, and V.~A. Zagrebnov}, {\em Disordered
  Bose-Einstein condensates with interaction in one dimension}, 
  J.~Stat.  Mech. {\bf 2012} (2012), P11007.

\bibitem{Solovej-ESI-2014}
{J.~P. Solovej}, {\em Many body quantum mechanics},
 Lecture notes at the Erwin Schroedinger Institute 2014, available
  online at
  \url{http://www.math.ku.dk/~solovej/MANYBODY/mbnotes-ptn-5-3-14.pdf}.

\bibitem{Spohn-80}
{H.~Spohn}, {\em Kinetic equations from {H}amiltonian dynamics: {M}arkovian
  limits}, Rev. Modern Phys. {\bf 52} (1980), 569--615.

\end{thebibliography}
\end{document}